\documentclass[final,twocolumn]{IEEEtran}

\usepackage{cite}
\usepackage{amsmath,amssymb,amsfonts}
\usepackage{faktor}
\usepackage{bm}
\usepackage{algorithmic}
\usepackage{graphicx}
\usepackage{textcomp}
\usepackage{xcolor}
\usepackage{rotating, booktabs}
\usepackage{wallpaper}
\usepackage{amsmath}
\usepackage{amsthm}
\usepackage{amsfonts}
\usepackage{amssymb}
\usepackage{graphicx}
\usepackage{float}
\usepackage{subfigure}
\usepackage{makecell}
\usepackage{hyperref}
\usepackage{indentfirst}
\usepackage{comment}
\usepackage{mathrsfs}
\usepackage{diagbox}
\usepackage{enumerate}
\usepackage{cuted} 

\usepackage[flushleft]{threeparttable}

\newtheorem{definitionenv}{Definition}
\newtheorem{lemmaenv}[definitionenv]{Lemma}
\newtheorem{theoremenv}[definitionenv]{Theorem}
\newtheorem{corollaryenv}[definitionenv]{Corollary}
\newtheorem{propositionenv}[definitionenv]{Proposition}
\newtheorem{conjectureenv}[definitionenv]{Conjecture}
\newtheorem{remarkenv}[definitionenv]{Remark}
\newenvironment{remark}{\begin{remarkenv}\rm}{\end{remarkenv}}
\newcommand{\er}{\end{remark}}

\newtheorem{exampleenv}[definitionenv]{Example}
\newtheorem{app-lemmaenv}[section]{Lemma}

\newenvironment{definition}{\begin{definitionenv}\rm}{\end{definitionenv}}
\newenvironment{lemma}{\begin{lemmaenv}\rm}{\end{lemmaenv}}
\newenvironment{theorem}{\begin{theoremenv}\rm}{\end{theoremenv}}
\newenvironment{corollary}{\begin{corollaryenv}\rm}{\end{corollaryenv}}
\newenvironment{example}{\begin{exampleenv}\rm}{\end{exampleenv}}

\newenvironment{app-lemma}{\begin{app-lemmaenv}\rm}{\end{app-lemmaenv}}

\theoremstyle{definition}

\newcommand{\cD}{{\mathcal D}}
\newcommand{\cG}{{\mathcal G}}

\newcommand{\cE}{{\mathcal E}}
\newcommand{\cP}{{\mathcal P}}
\newcommand{\cQ}{{\mathcal Q}}
\newcommand{\cW}{{\mathcal W}}

\newcommand{\mF}{{\mathbb F}}

\newcommand{\PQ}{P}
\newcommand{\hPQ}{\hat{P}}

\newcommand{\wt}[1]{{\mathrm{wt}\left(#1\right)}}

\usepackage[colorinlistoftodos,prependcaption]{todonotes}
\usepackage{xcolor}
\usepackage{xargs}

\newcommandx{\rednote}[2][1=]{\todo[inline,linecolor=red,backgroundcolor=red!25,bordercolor=red,#1]{#2}}
\newcommandx{\yellownote}[2][1=]{\todo[inline,linecolor=yellow,backgroundcolor=yellow!25,bordercolor=yellow,#1]{#2}}

\newcommand{\ket}[1]{\left|{#1}\right\rangle}
\newcommand{\bra}[1]{\left\langle{#1}\right|}

\begin{document}

\title{Semidefinite programming bounds on the size of entanglement-assisted codeword stabilized quantum codes}
\author{Ching-Yi Lai, Pin-Chieh Tseng, and Wei-Hsuan Yu
\thanks{\footnotesize
This article was presented in part at the 2024 IEEE International Symposium on Information Theory (ISIT).

CYL and PCT were supported by the National Science and Technology Council (NSTC) in
Taiwan, under Grant 111-2628-E-A49-024-MY2, 112-2119-M-A49-007, and 112-2119-M-001-007. Yu is partially supported by the the National Science and Technology Council, under Grant No. 109-2628-M-008-002-MY4. 

Ching-Yi Lai  and Pin-Chieh Tseng are with the Institute of Communications Engineering, National Yang Ming Chiao Tung University (NYCU), Hsinchu 30010, Taiwan. Pin-Chieh Tseng is also with the Department of Applied Mathematics, NYCU. (emails: cylai@nycu.edu.tw and pichtseng@gmail.com)
			
Wei-Hsuan Yu is with the Department of Mathematics, National Central University, Taoyuan 32001, Taiwan. (email: u690604@gmail.com)
}}

\date{\today}

\maketitle
	
\begin{abstract}

In this paper, we explore the application of semidefinite programming  to the realm of quantum codes, specifically focusing on codeword stabilized (CWS) codes with entanglement assistance. Notably, we utilize the isotropic subgroup of the CWS group and the set of word operators  of a CWS-type quantum code to derive an upper bound on the minimum distance. Furthermore, this characterization can be incorporated into the  associated distance enumerators, enabling us to construct semidefinite constraints that lead to SDP bounds on the minimum distance or size of CWS-type quantum codes. We illustrate several instances where SDP bounds outperform LP bounds, and there are even cases where LP fails to yield meaningful results, while SDP consistently provides tighter and relevant bounds. Finally, we also provide interpretations of the Shor-Laflamme weight enumerators and shadow enumerators  for codeword stabilized codes, enhancing our understanding of quantum codes.

\end{abstract}

\section{Introduction}

Quantum states are inherently vulnerable to errors. To protect quantum information, error correction techniques are employed \cite{Shor95,CS96,Ste96,Got97,CRSS98}. These techniques use multiple physical qubits  to represent a small number of logical qubits, allowing for the recovery of the logical state when errors occur, provided they are not too severe. In a broader sense, an $((n, M))$ quantum code is defined as an $M$-dimensional subspace $\cQ$ of the $n$-qubit state space $\mathbb{C}^{2^n}$, and quantum codewords are unit vectors within this subspace. Error operators acting on quantum codewords can be refined to the class of Pauli operators that act on $\mathbb{C}^{2^n}$. The minimum distance of a quantum code is defined as the minimum weight of an undetectable Pauli error, based on the quantum error correction conditions \cite{KL97,BDSW96}.

One essential question in coding theory is to determine the maximum size $M$ for given parameters. 
The quantum Gilbert-Varshamov bound  proves the existence of good quantum codes~\cite{FM04} and sometimes these can be achieved by code concatenation  \cite{Ouy14}. One can formulate the problem of estimating the maximum code size as a convex optimization problem, specifically as a linear programming (LP) or semidefinite programming (SDP) problem. The goal is to determine an upper bound on the size of a code, which is referred to as an LP or SDP bound~\cite{Del73,Sch05}.
Usually,   LP bounds provide the strongest known bounds on code sizes for quantum codes~\cite{CRSS98,AL99,LBW13,LA18,ALB20}, outperforming the quantum Hamming bounds and quantum Singleton bounds \cite{Got96,CC97,HG20,GHW22}.  LP bounds  also apply to non-Pauli errors~\cite{OL22}.

Delsarte initially derived linear inequalities pertaining to code distance enumerators~\cite{Del73}. Notably,  these linear constraints can be viewed as corollaries of the MacWilliams identities in the Hamming scheme \cite{MS77}.  Furthermore, the introduction of shadow enumerators  and shadow identities for self-dual codes~\cite{CS90},
along with split weight enumerators and split MacWilliams identities~\cite{Sim95}, has provided  additional linear constraints. 
These concepts have also been extended to the realm of quantum codes. 
Shor and Laflamme introduced weight enumerators for general quantum codes and derived the quantum version of the MacWilliams identities~\cite{SL97}. Rains defined the quantum shadow enumerators for general quantum codes and proved the shadow identities~\cite{Rain98,R99}. Lai and Ashikhmin explored the split weight enumerators and split MacWilliams identities in~\cite{LA18}.

On the other hand, semidefinite constraints for binary codes were first introduced by Schrijver, relying on the Terwilliger algebra within the Hamming or Johnson schemes \cite{Sch05}.  Additional semidefinite constraints can be generated from distance relations among multiple points \cite{GMS12} and split Terwilliger algebras \cite{TLY23}. In general, SDP bounds offer improvements over their corresponding LP bounds, particularly for nonlinear codes. For example, SDP bounds improve the maximum size of binary nonlinear codes with $n = 25, d = 6$ to $47998$ \cite{Sch05}, where LP bounds give us $48148$. SDP has matrix inequalities where LP does not have. Therefor, SDP might improve LP bounds. This methodology can also be extended to non-binary codes \cite{GST06}.

In this paper, we extend the SDP approach to entanglement-assisted (EA) quantum codes~\cite{Bow02,BDH06}.
The challenge arises from the differing domains between classical and quantum codes. Classical codes are sets of vectors over a finite field, while quantum codes exist in complex Hilbert spaces.
Semidefinite constraints for classical codes are derived from  distance relations among triplets of codewords \cite{Sch05, GST06}. 
Typically, the weight enumerators introduced by Shor and Laflamme serve to characterize the minimum distance of a quantum code. However, these weight enumerators do not have  comprehensive combinatorial interpretations~\cite{NK22} 
 and this limits direct application of  SDP techniques.  
 
We specifically consider the class of codeword stabilized (CWS) codes  with entanglement assistance (EA) \cite{CSSZ09,SHB11}, which can be defined by a CWS group and a set of word operators.
More details about CWS code constructions and their properties  can be found in \cite{SSW07,YCLO08,CCSSZ09,ROJ19,Kap21,KT23}.
In general, the associated classical codes of these (EA) CWS codes are non-additive codes, and determining the minimum distance of a CWS code is a challenging problem~\cite{CSSZ09,KDP11}. We revisit the quantum error correction conditions \cite{KL97,BDSW96} and study how the minimum distance of a quantum code can be characterized. We introduce the concept of stabilizer sets of Pauli operators for quantum codes with various eigenvalues, which extends the notion of stabilizer codes. We demonstrate that for CWS-type quantum codes, the only possible generalized stabilizer sets have eigenvalues $0$ and $\pm 1$, indicating that these operators are either detectable or have no detrimental impact on the codespace. Consequently, the minimum distance of a CWS-type quantum code corresponds to the minimum weight of a Pauli operator outside of these generalized stabilizer sets.

We further illustrate how upper bounds on the minimum distance of a CWS-type quantum code can be established by selecting a set of Pauli errors that are undetectable. Notably, the minimum weight of an element within a set of undetectable errors by a quantum code  serves as an upper bound on the minimum distance of the code. In particular, we can utilize the isotropic subgroup of the CWS group of a CWS-type quantum code and the set of word operators to derive an upper bound on the  minimum distance. Furthermore, this characterization can be incorporated into the distance enumerators associated with the isotropic subgroup and the set of word operators, enabling us to construct semidefinite constraints  that provide SDP bounds on the minimum distance or size of CWS-type quantum  codes.

In addition to the newly derived semidefinite constraints, we also take into account split MacWilliams identities and split shadow identities for EA-CWS codes. To facilitate a comprehensive comparison, we begin by constructing a table of upper bounds with only the linear constraints. We then combine the linear constraints with semidefinite constraints to create SDP bounds. The results are summarized in Table \ref{table:sdp}, which clearly illustrates several instances where SDP bounds outperform LP bounds, showcasing the advantages of semidefinite programming in this context.
In numerous cases, LP fails to yield meaningful results, whereas SDP consistently provides {tighter} and relevant bounds. 
As a byproduct, we construct a table of upper bounds for EA stabilizer codes in Table~\ref{table:stabilizer}, which enhances Grassl's table~\cite{Gra_tab}.

Furthermore, we delve into the interpretation of the $A$-type weight enumerators introduced by Shor and Laflamme 
and the quantum shadow enumerators for CWS codes 
 because  the interpretation of the weight enumerators is important in the application of linear and semidefinite programming.
We show that the $A$-type weight enumerator characterizes  the commutation relations between the elements in the CWS group and the word operators of a CWS code.

Finally, we also provide an  interpretation of the shadow enumerators
as certain distance enumerators between the set of word operators times the CWS group to one of its {cosets}, further enhancing our understanding of the weight enumerators and their relevance to quantum codes.

The paper is organized as follows. In the next section, we briefly review the SDP method for classical nonbinary codes and the theory of quantum codes. 
In Section~\ref{sec:CWS}, we review EA-CWS codes and study their fundamental properties.
In Section~\ref{sec:code_structure}, we delve into code structure by defining generalized stabilizer sets and exploring their relations to the minimum distance of quantum codes.  The semidefinite constraints for EA-CWS codes are given in Section~\ref{sec:SDP}. 
In Section~\ref{sec:app}, we demonstrate the applications of the semidefinite constraints to EA-CWS codes.  
In Section~\ref{sec:quantum_enum}, we provide  interpretations for the $A$-type weight enumerators and the quantum shadow enumerators of CWS codes. Finally, we present conclusions in Section~\ref{sec:con}.

\section{Preliminaries}
\label{sec:pre}

\subsection{SDP   for non-binary codes}
\label{sec:SDP_classical}
Consider an $n$-dimensional vector space $\mF_{q}^n$ over a finite field $\mF_{q}$. The support of a vector $v=(v_1,\dots,v_n)\in\mF_{q}^n$ is defined as
\begin{equation*}
    {\rm Supp}(v) = \{i : v_{i} \neq 0\}.
\end{equation*}
The weight of $v$, denoted $\wt{v}$, is the size of ${\rm Supp}(v)$, which is the number of its nonzero entries $v_i$.
The distance between two vectors 
 $v = (v_{1}, \dots, v_{n})$ and $w = (w_{1}, \dots, w_{n}) \in\mF_q^n$, denoted $d(v,w)$, is defined as 
  $d(v,w)=\wt{v-w} $.

An $(n, M, d)$ code $C$ over  $\mathbb{F}_q$ is defined as a subset of $\mathbb{F}_q^n$ that contains $M$ vectors and the distance between any two vectors in $C$ is at least $d$. Each vector in $C$ is referred to as a codeword.

In the following we introduce basics of the SDP method proposed by Schrijver~\cite{Sch05,GST06,Lau07}.
We will extend this method for quantum CWS codes in Section~\ref{sec:SDP}.
   
Let ${\mathcal S}_{m}$ be the symmetric group of $m$ elements.   Consider the group action of $\cG = \big(\prod_{i = 1}^{n} {\mathcal S}_{q-1}\big) \times {\mathcal S}_{n}$ on $\mF_{q}^n$ as follows. 
For $v = (v_{1}, \dots, v_{n}) \in \mF_{q}^n$, and $g= (\pi_{1}, \dots, \pi_{n+1}) \in \cG$, with $\pi_{i} \in {\mathcal S}_{q-1}$ for $i< n+1$ and $\pi_{n+1} \in {\mathcal S}_{n}$,  we have
\begin{equation*}
    g(v) = (\pi_{1}(v_{\pi_{n+1}(1)}), \dots, \pi_{n}(v_{\pi_{n+1}(n)})).
\end{equation*}
Note that the elements in $\mathcal{S}_{q-1}$ represent permutations on $\mF_q\setminus \{0\}$.
The action of $\cG$ on $\mF_{q}^n \times \mF_{q}^n$  can be similarly defined as 
\begin{equation*}
    g(v, w) = (g(v), g(w))
\end{equation*}
for $g \in \cG$, and $v, w \in \mF_{q}^n$. Observe that $(v, w)$ and $(v', w') \in \mF_{q}^n \times \mF_{q}^n$ have the same orbit under the action of $\cG$ if and only if 
\begin{align}
    &{\rm wt}(v) = {\rm wt}(v'),\\
    &{\rm wt}(w) = {\rm wt}(w'),\\  
    &\lvert {\rm Supp}(v) \cap {\rm Supp}(w) \rvert = \lvert {\rm Supp}(v') \cap {\rm Supp}(w') \rvert,\\
    &d(v, w) = d(v', w').
\end{align}
{This relation can be regarded as a description of the distances between $\emptyset, v$ and $w$.} Consequently,  $q^{n} \times q^{n}$ matrices $M_{i, j}^{t, p}$ for $1\leq i, j\leq n$, $0 \leq p \leq t \leq {\rm min}\{i, j\}$ and $i+j \leq n+t$ are defined to characterize the pairwise relations between any two vectors in $\mF_q^n$.  The matrix $M_{i, j}^{t, p}$ is indexed by the elements of $\mF_{q}^n$ 
and its $(v,w)$-th entry  for $v,w\in \mF_{q}^n$ is
\begin{equation}  \label{eq:Mijtp}
    \left(M_{i, j}^{t, p}\right)_{v, w} = \left\{
    \begin{aligned}
        &1, & &\text{if } {\rm wt}\left(v\right) = i, {\rm wt}\left(w\right) = j,\\
        & & &\lvert {\rm Supp}(v) \cap {\rm Supp}(w) \rvert = t, \\
        & & &\lvert {\rm Supp}(v) \setminus {\rm Supp}(v-w) \rvert = p;\\
        &0, & &\text{otherwise.}
    \end{aligned}
    \right.
\end{equation}
Consider that ${\rm Supp}(v) \cap {\rm Supp}(w)$ represents the collection of positions where both $v$ and $w$ contain non-zero values. Note that in the case of non-binary alphabets, these non-zero values can differ. Consequently, we also have to consider $ {\rm Supp}(v) \setminus {\rm Supp}(v-w)$, which is  the set of positions where $v$ and $w$ share common non-zero values.

Let ${\mathcal A}_{n, q}$ be the algebra generated by the matrices $M_{i, j}^{t, p}$. Then, ${\mathcal A}_{n, q}$ is the Terwilliger algebra of the Hamming scheme ${\mathcal H}(n, q)$~\cite{GST06}.

On the other hand, one can define positive semidefinite matrices by the distance preserving automorphisms for $C$. 
For $v \in C$,  define matrices $R_{C}^{v}$, indexed by the elements of $\mF_{q}^n$, as
\begin{equation*}
    R_{C}^{v} = \sum_{g \in \cG} \chi^{g(\pi_{v}(C))} \left(\chi^{g(\pi_{v}(C))}\right)^{T},
\end{equation*}
where $\pi_{v}$ is a distance preserving automorphism for $C$, that maps $v$ to the zero vector, and $\chi^{g(\pi_{v}(C))}$ is defined as a column vector indexed by the elements of $\mF_{q}^n$ with
\begin{equation*}
    \left(\chi^{g(\pi_{v}(C))}\right)_{i} = \left\{
    \begin{aligned}
        &1, & &\text{if } i \in g(\pi_{v}(C));\\
        &0, & &\text{otherwise.}
    \end{aligned}
    \right.
\end{equation*}
It is clear that $R_{C}^{v}$ is positive semidefinite.
Then, one can define two  positive semidefinite  matrices
\begin{align*}
    &R_{C} = \left\lvert C \right\rvert^{-1} \sum_{v \in C} R_{C}^{v},\\
    &R_{C}' = \left\lvert q^{n} - \lvert C \rvert \right\rvert^{-1} \sum_{v \in \mF_{q}^n \setminus C} R_{C}^{v}.
\end{align*}

 Moreover, these two matrices lie in the matrix space generated by $\{M_{i,j}^{t,p}\}$.
 
\begin{theorem}\cite{GST06}
\label{thm:sdp_matrix}
Consider
\begin{align*}
    \lambda_{i, j}^{t, p} = \Big\lvert &\{(u, v, w) \in C \times C \times C :\wt{u-v} = i, 
    \nonumber\\
    & \wt{u-w} = j, \lvert {\rm Supp}(u-v) \cap {\rm Supp}(u-w) \big\rvert = t, \nonumber\\
    &
        \lvert {\rm Supp}(u-v) \setminus {\rm Supp}(v-w) \rvert = p\}\Big\rvert.
\end{align*}
Let
\begin{equation*}
    x_{i, j}^{t, p} = \frac{1}{\lvert C \rvert \gamma_{i, j}^{t, p}} \lambda_{i, j}^{t, p}, \label{eq:xijtp}
\end{equation*}
where
\begin{equation*}
    \gamma_{i, j}^{t, p} = (q-1)^{i+j-t} (q-2)^{t-p}\binom{n}{p, t-p, i-t, j-t}.\label{eq:gijtp}
\end{equation*}
Then
\begin{align*}
    &R_{C} = \sum_{i, j, t, p} x_{i, j}^{t, p} M_{i, j}^{t, p},\\
    &R_{C}' = \frac{\lvert C \rvert}{q^{n} - \lvert C \rvert}\sum_{i, j, t, p} (x_{i+j-t-p, 0}^{0, 0}- x_{i, j}^{t, p}) M_{i, j}^{t, p}.
\end{align*}

\end{theorem}

To sum up, we have the following semidefinite constraints on the distance relations $x_{i,j}^{t,p}$ for a classical code.
\begin{theorem}\cite[SDP  for non-binary codes]{GST06} \label{thm:SDP_constraints}
Let $C$ be a classical code over $F_{q}$ and  $x_{i, j}^{t, p}$  be the distance relations for $C$ as defined  in ~(\ref{eq:xijtp}).
Then,  $x_{i, j}^{t, p}$ satisfy the following constraints:
\begin{enumerate}
    \item
    $R_{C}$ and $R'_{C}$ are positive semidefinite;
    \item 
    $x_{0, 0}^{0, 0} = 1$;
    \item 
    $0 \leq x_{i, j}^{t, p} \leq x_{i, 0}^{0, 0}$;
    \item 
    $x_{i, j}^{t, p} = x_{i', j'}^{t', p'}$ if $t-p = t'-p'$ and $(i, j, i+j-t-p)$ is a permutation of $(i', j', i'+j'-t'-p')$;
    \item 
    $x_{i, j}^{t, p} = 0$ if $\{i, j, i+j-t-p\}\cap \{1, \dots, d-1\} \neq \emptyset$.
\end{enumerate}
Moreover, we have
\[
\lvert C \rvert = \sum_{i} \gamma_{i, j}^{t, p}x_{i, j}^{t, p}.
\]

\end{theorem}
 
\subsection{Quantum codes}

In this context, we consider errors that  are tensor product of  Pauli matrices 
\begin{align*}
I = \begin{pmatrix}
    1 & 0\\
    0 & 1
\end{pmatrix}, &
\quad
\sigma_{x} = \begin{pmatrix}
    0 & 1\\
    1 & 0
\end{pmatrix},\\
\sigma_{y} = \begin{pmatrix}
    0 & - \bm{i} \\
    - \bm{i} & 0
\end{pmatrix}, &
\quad
\sigma_{z} = \begin{pmatrix}
    1 & 0 \\
    0 & -1
\end{pmatrix},
\end{align*}
in the computational basis $\{\ket{0},\ket{1}\}$, where $\bm{i} = \sqrt{-1}$. 

An error occurring on an $n$-qubit state    is defined as the action of an  $n$-fold Pauli operator $\alpha$ of the form $\sigma_{1} \otimes \cdots \otimes \sigma_{n}$ on  the $n$-qubit Hilbert space $\mathbb{C}^{2^{n}}$, where $\sigma_{i}$'s are Pauli matrices. 
We will refer to the Pauli operator $\alpha$ as the error in the following.
To analyze the error correction capability of quantum codes, it is sufficient to consider the following set $ {\cE}^{n}$:
\begin{equation*}
     {\cE}^{n}= \{\sigma_{1} \otimes \cdots \otimes \sigma_{n}: \sigma_{i} \in\{I, \sigma_{x},\sigma_{y},\sigma_{z}\} \}.
\end{equation*}
This set can be extended to the $n$-fold Pauli group   $\cP^{n}=\{ \bm{i}^k \alpha :\ k\in\{0,1,2,3\},  \alpha\in  {\cE}^{n} \}$.
It is important to mention that two Pauli operators either commute or anticommute. 

The support of a Pauli operator $\bm{i}^{k}\alpha = \bm{i}^{k} \sigma_{1} \otimes \cdots \otimes \sigma_{n} \in \cP^{n}$ is defined as  ${\rm Supp}(\bm{i}^{k}\alpha)=\{j : \sigma_{j} \neq I\}$, which is the index set of its nontrivial Pauli components. Then, the weight of $\bm{i}^{k}\alpha$, denoted ${\rm wt}(\bm{i}^{k}\alpha)$, is defined as the size of  its support.
In addition, we define the distance between two Pauli operators $\alpha$ and $\beta\in\cP^n$
as $\wt{\alpha\beta}$.
We will count the distance distributions of subsets of $\cP^n$. Since the phase of a Pauli operator is irrelevant to its weight, our counting can be done over $\cE^n$. Thus, we define a mapping 
 $\varphi:\cP^{n}\rightarrow \cE^{n}$ by 
 \begin{align}
     \varphi\left(\bm{i}^{k}\alpha\right) =\alpha \label{eq:homo}
 \end{align}
  for $\alpha\in\cE^n$.
For a set $T\subset\cE^n$, its weight enumerator is 
\begin{align}
    W_i=\big|  \{  \alpha\in T: \wt{\alpha}=i \}\big|, \ i=0,1,\dots, n,
\end{align}
and its distance enumerator is
\begin{align}
    W_i'= \frac{1}{|T|}\big|  \{  \alpha,\beta \in T: \wt{\alpha\beta}=i \}\big|, \ i=0,1,\dots, n.
\end{align}
The distance and weight enumerators of a group coincide.

A quantum code   of length $n$ and size $M$ is defined as an $M$-dimensional subspace of  $\mathbb{C}^{2^{n}}$. 
The  error correction conditions have been formulated in~\cite{KL97,BDSW96}, which also imply the conditions that errors can be detected by  a quantum code. 
     Let $\mathcal{Q}$ be a quantum code of length $n$.  
     An error $\alpha \in \cE^n$ is  detectable by $\mathcal{Q}$ if and only if for any basis states $|v_i\rangle, |v_j\rangle \in \mathcal{Q}$,  we have
\begin{equation}
\label{error_con_1}
    \langle v_i | \alpha | v_i \rangle = \langle v_j | \alpha | v_j \rangle,
\end{equation}
and 
\begin{equation}
    \label{error_con_2}
    \langle v_i | \alpha | v_j \rangle = 0 \text{ if } \langle v_i | v_j \rangle = 0.
\end{equation}
The quantum code  $\mathcal{Q}$   is said to have minimum distance $d$ if any $\alpha \in \cE^{n}$ of weight up to $d-1$, but not $d$, are detectable by $\cQ$. This type of quantum code is   denoted as an $((n, M, d))$ code or $((n,M))$ code if  the minimum distance is not specified.

Next, we introduce  weight enumerators for quantum codes \cite{SL97,R99}. 
Let   
\begin{equation*}
    \mathcal{E}_{j}^{n} = \{\beta \in {\cE}^{n} : {\rm wt}(\beta) = j\}
\end{equation*}
denote the set of $n$-fold Pauli errors of weight $j$. 
Shor and Laflamme defined $A$- and $B$- types of weight enumerators for a quantum code as follows~\cite{SL97}.
\begin{definition} \label{def:SL}
 Let $\mathcal{Q}$ be an $((n, M))$ quantum code and $\PQ$ be the projection operator onto $\mathcal{Q}$. The Shor-Laflamme (SL) weight enumerators for $\mathcal{Q}$ are
\begin{align*}
    A_{j} =& \frac{1}{ \left( {\rm Tr}(\PQ) \right)^2 }\sum_{\beta \in \mathcal{E}_{j}^{n}}{\rm Tr}(\beta \PQ) {\rm Tr}(\beta \PQ),\\
    B_{j} =& \frac{1}{{\rm Tr}(\PQ)}\sum_{\beta \in \mathcal{E}_{j}^{n}}{\rm Tr}(\beta \PQ \beta \PQ),
\end{align*}
 for $j = 0,1, \dots, n$.
\end{definition}
\begin{lemma}\cite{SL97} \label{lemma:SL_enum}
Let $\mathcal{Q}$ be an $((n, M, d))$ quantum code and $\PQ$ be the projection operator onto $\mathcal{Q}$.
    Suppose that $\PQ = \sum_{i} |v_{i}\rangle \langle v_{i} |$, where $\ket{v_i}$ are orthonormal basis states for $\cQ$. Then, the SL weight enumerators are 
    \begin{equation}
    \label{eq:SL_A}
        A_{j} = \frac{1}{M^{2}}\sum_{\beta \in \mathcal{E}_{j}^{n}} \left\lvert\sum_{i}  \langle v_{i} | \beta | v_{i} \rangle \right\rvert^{2},
    \end{equation}
    \begin{equation}
    \label{eq:SL_B}
        B_{j} = \frac{1}{M} \sum_{\beta \in \mathcal{E}_{j}^{n}}\sum_{i, j} \left\lvert \langle v_{i} | \beta | v_{j} \rangle \right\rvert^{2},
    \end{equation}
    for $j = 0,1, \dots, n$.
    Moreover,  $A_{j} = B_{j}$ for $j\leq d-1$, $A_d<B_d$, and $A_{j} \leq B_{j}$ for all $j$.
\end{lemma} 
The two  SL weight enumerators $\{A_{j}\}$ and $\{B_{j}\}$ satisfy the MacWilliams identities as follows.
\begin{lemma}\cite{SL97}
\label{lem:gen_lp}
Let $\{A_{j}\}$ and $\{B_{j}\}$ be the    SL weight enumerators for an $((n,M))$ code. Then,
for $j=0,1,\dots,n$,
    \begin{equation*}
        B_{j} = \frac{M}{2^{n}} \sum_{i=0}^n K_{j}(i, n) A_{i},
    \end{equation*}
    where
    \begin{equation}
        K_{i}(k, n) = \sum_{j = 1}^{i} (-1)^{j}3^{i-j}\binom{k}{j} \binom{n-k}{i-j} \label{eq:kraw}
    \end{equation}
    is the $4$-ary Kravchuk polynomial.
\end{lemma}
Now, we introduce  quantum shadow enumerators \cite{R99}.
\begin{definition} \label{def:shadow}
Let $\mathcal{Q}$ be an $((n, M))$ quantum code and $\PQ$ be the projection operator onto $\mathcal{Q}$. Define
{
    \begin{equation*}
        \hPQ = \frac{1}{2^{n}}\sum_{\alpha \in {\cE}^{n}} (-1)^{{\rm wt}(\alpha)} {\rm Tr}(\alpha^{\dagger}\PQ) \alpha.
    \end{equation*}
}
 $\hPQ $ and $P$ satisfy that
    \begin{equation}
        \hPQ =  \sigma_y^{\otimes n } \bar{P}~\sigma_y^{\otimes n }, \label{eq:Phat}
    \end{equation}
   where $\bar{P}$ is the complex conjugate of $P$,
   which follows from the relations $\bar{I}=I, \bar{\sigma_x}=\sigma_x, \bar{\sigma_y}=-\sigma_y,$ and $ \bar{\sigma_z}=\sigma_z$.
    Then, the shadow enumerator for $\mathcal{Q}$ is  defined as 
    \begin{equation*}
        D_{j} = \left\{
        \begin{aligned}
        &\frac{1}{{\rm Tr}(\PQ\hPQ)}\sum_{\beta \in \mathcal{E}_{j}^{n}}{\rm Tr}(\beta \PQ \beta \hPQ), & &\text{if } {\rm Tr}(\PQ\hPQ) \neq 0;\\
        & \sum_{\beta \in \mathcal{E}_{j}^{n}}{\rm Tr}(\beta \PQ \beta \hPQ), & &\text{otherwise,}
        \end{aligned}\right.
    \end{equation*}
    for $j=0,1, \dots, n$.
\end{definition}

The quantum shadow enumerator of a quantum code exhibits a linear relationship   with the $A$-type SL weight enumerator.
\begin{lemma}\cite{R99}
\label{lem:shadow}
Let $\mathcal{Q}$ be an $((n, M))$ quantum code and $\PQ$ be the projection operator onto $\mathcal{Q}$.
Let $\{A_j\}$ and $\{D_j\}$ be the $A$-type SL and shadow enumerators for $\cQ$, respectively. Then, 
    \begin{equation} 
        D_{j} = \left\{
        \begin{aligned}
            &\frac{M}{2^{n}}\sum_{i=1}^n(-1)^{i}K_{j}(i, n)A_{i}, & &\text{if } {\rm Tr}(\PQ\hPQ) \neq 0,\\
            &\frac{M^{2}}{2^{n}}\sum_{i=1}^n(-1)^{i}K_{j}(i, n)A_{i}, & &\text{if } {\rm Tr}(\PQ\hPQ) = 0,\\
        \end{aligned}\right. \label{eq:Sh_identity}
    \end{equation}
     for $j = 0, \dots, n$. Moreover, $D_{j} \geq 0$ for all $j$.
\end{lemma}
Equation~(\ref{eq:Sh_identity}) is referred to as the \textit{quantum shadow identities}.

\subsection{Entanglement-assisted quantum codes}
In this section, we introduce entanglement-assisted (EA) quantum codes
in the context of quantum communication,
where entanglement is shared between the sender and receiver~\cite{BDH06}.

\begin{definition}
    An $((n, M, d; c))$ EA quantum code $\mathcal{Q}$ is  an $M$-dimensional subspace of $\mathbb{C}^{2^{n+c}} = \mathbb{C}^{2^{n}} \otimes \mathbb{C}^{2^{c}}$ such that the following two conditions hold:
    \begin{enumerate}
        \item 
        for $|v \rangle \in \mathcal{Q}$, ${\rm Tr}_{\mathbb{C}^{2^{n}}}(|v\rangle \langle v|) = \frac{I^{\otimes c}}{ 2^{c} }$, where the partial trace is performed with respect to the first $n$ qubits/
        \item 
        for $\alpha \in \cE^{n}$ of  weight up to $d-1$, but not $d$,   $\alpha \otimes I^{\otimes c}$ is detectable by $\cQ$.
    \end{enumerate}
\end{definition}

An $((n, M, d; c))$ EA quantum code encodes an $M$-dimensional  quantum state  into $n+c$ physical qubits. The first $n$ qubits are held by the sender, while the remaining $c$ qubits are with the receiver. {The sender sends the $n$-qubit to the receiver through a noisy quantum channel, and the  $c$ qubits with the receiver are assumed to be error-free.} Our analysis primarily concentrates on the error correction capabilities of the first $n$ qubits. Thus, it is convenient to perform weight enumeration separately for the first $n$ qubits and the remaining $c$ qubits \cite{LA18}. 
The following definition generalizes Defs.~\ref{def:SL}  and~\ref{def:shadow}. 

\begin{definition}
Let $\mathcal{Q}$ be an $((n, M, d; c))$ EA quantum code and $\PQ$ be the projection operator onto $\mathcal{Q}$. Then, the split  SL and shadow weight enumerators $\{A_{i, j}\}$, $\{B_{i, j}\}$, and $\{D_{i,j}\}$ are defined by
    \begin{align}
        A_{i, j} =& \frac{1}{M^{2}}\sum_{\alpha \in \cE_{i}^{n}, \beta \in \mathcal{E}_{j}^{c}}{\rm Tr}((\alpha \otimes \beta) \PQ) {\rm Tr}((\alpha \otimes \beta) \PQ),\\
        B_{i, j} =& \frac{1}{M} \sum_{ \alpha \in \cE_{i}^{n}, \beta \in \mathcal{E}_{j}^{c}}{\rm Tr}((\alpha \otimes \beta) \PQ (\alpha \otimes \beta) \PQ),\\
        D_{i,j} =& \begin{cases}\displaystyle
        \frac{B_{i, j}(\PQ, \hPQ)}{{\rm Tr}(\PQ\hPQ)}, & \text{if } {\rm Tr}(\PQ\hPQ) \neq 0;\\
        B_{i, j}(\PQ, \hPQ), & \text{otherwise;}
        \end{cases} 
    \end{align}
for $i=0,1,\dots, n$ and $j=0,1, \dots, c$ with
\begin{equation*}
    B_{i, j}(\PQ, \hPQ) = \sum_{\alpha \in \cE_{i}^{n}, \beta \in \mathcal{E}_{j}^{c}}{\rm Tr}((\alpha \otimes \beta) \PQ (\alpha \otimes \beta) \hPQ).
\end{equation*}
\end{definition}
From the definition of  split weight enumerators, we can observe the relationship between the split weight enumerators $\{A_{i, j}\}$, $\{B_{i, j}\}$, and $\{D_{i, j}\}$ and the standard weight enumerators $\{A_{j}\}$, $\{B_{j}\}$, $\{D_{j}\}$ for $\cQ$:
\begin{equation}
\label{equ:id_split}
    A_{i} = \sum_{j+k = i}A_{j, k},   B_{i} = \sum_{j+k = i} B_{j, k}, \text{ and } D_{i} = \sum_{j+k = i} D_{j, k},
\end{equation}
for $i= 0,1, \dots, n$. Moreover,   the following properties hold.

\begin{lemma}
\label{lem:EA_split_lp}
    Let $\mathcal{Q}$ be an $((n, M, d; c))$ entanglement-assisted quantum code with   split weight enumerators $\{A_{i, j}\}$ and $\{B_{i, j}\}$. Then, we have
    \begin{enumerate}
        \item 
        $A_{0, 0} = B_{0, 0} = 1$, $B_{i, j} \geq A_{i, j} \geq 0$ and $D_{i,j}\geq 0$ for all $i, j$;
        \item 
        $A_{0, j} = 0$ for $j=1, \dots, c$;
        \item 
        $A_{i, 0} = B_{i, 0}$ for $i=0, \dots, d-1$ and $A_{d, 0} < B_{d, 0}$;
        \item  (split MacWilliams identities)
        $B_{i, j} = \frac{M}{2^{n+c}} \sum_{u = 0}^{n} \sum_{v = 0}^{c} K_{i}(u, n) K_{j}(v, c) A_{u, v} \text{ for all } i, j$;
        \item  (split shadow identities) 
        \begin{equation}
            \sum_{u = 0}^{n} \sum_{v = 0}^{c} (-1)^{u+v} K_{i}(u, n) K_{j}(v, c)A_{u,v} \geq 0.
            \label{eq:split_Sh_identity}
        \end{equation}
    \end{enumerate}
\end{lemma}

\begin{proof}
    The conditions 1) to 4) are from \cite[Theorem $2$]{LA18}. By generalizing the proof of Lemma~\ref{lem:shadow}, we have
    \begin{equation}
        D_{i,j} = 
        \left\{\begin{aligned}
            & \frac{M}{2^{n+c}} {\rm Sh}(i, j), &\text{if } {\rm Tr}(\PQ\hPQ) \neq 0;\\
            &\frac{M^{2}}{2^{n+c}} {\rm Sh}(i, j),  &\text{if } {\rm Tr}(\PQ\hPQ) = 0,\\
        \end{aligned}\right.
                    \label{eq:split_Sh_identity0}
    \end{equation}
    where 
    \begin{equation*}
        {\rm Sh}(i, j) = \sum_{u = 0}^{n} \sum_{v = 0}^{c} (-1)^{u+v} K_{i}(u, n) K_{j}(v, c)A_{u,v}.
    \end{equation*}
    Since  $D_{i,j}\geq 0$,  we have the inequality~(\ref{eq:split_Sh_identity}) in  5).
\end{proof}
 
 Lemma~\ref{lem:EA_split_lp} 4)  can be regarded as the quantum analogue of split MacWilliams identities for classical codes described in \cite{Sim95}. Moreover, when $c = 0$, the identities will collapse to the previous MacWilliams identities.
 Similarly, ~(\ref{eq:split_Sh_identity0}) is referred to as the split shadow identities.
 Its implication  (\ref{eq:split_Sh_identity}) imposes additional linear constraints on the split weight enumerator $\{A_{i,j}\}$.

\section{EA-CWS codes}
\label{sec:CWS}

In this section, we introduce (EA) CWS quantum codes \cite{CSSZ09,SHB11} and study their fundamental properties.

An abelian Pauli subgroup that does not contain the minus identity is called stabilizer group.
The \textit{stabilized space} of a stabilizer group is defined as the
collective joint eigenspace of all elements in the stabilizer group with eigenvalues of $+1$. 
If $S$ is a maximal stabilizer subgroup, then the stabilized space of $S$ is of dimension one and is called a stabilizer state, which will be denoted by $\ket{S}$. In this case, $S$ is called a \textit{CWS group}.
 \begin{definition}
      An $((n,M,d;c))$ EA quantum code $\cQ$ is classified as a CWS code if there exist a CWS group $S\subset \cP^{n+c}$   along with a set of  $M$ commuting operators $W=\{\omega_1=I^{\otimes n},\dots,\omega_M \}\subset \cP^{n}$
    such that  $$\big\{\ket{S}, \left(\omega_2\otimes I^{\otimes c}\right)\ket{S},\dots,\left(\omega_M\otimes I^{\otimes c}\right) \ket{S} \big\}$$ forms an orthonormal basis for $\cQ$. 
 \end{definition}
   The  operators $\{\omega_j\}_{j=1}^M$ for an EA-CWS code are called \textit{word operators}. 
    Moreover, each of the {basis states} $\left(\omega_j\otimes I^{\otimes c}\right)\ket{S}$ is a stabilizer state, stabilized by
    $\left(\omega_2\otimes I^{\otimes c}\right) S \left(\omega_2^\dag\otimes I^{\otimes c}\right)$. 
    In particular, one can show that these states can be prepared by applying Clifford gates~\cite{NC00}
    to the first $n$ qubits of the state $\ket{0}^{\otimes (n-c)} \otimes \ket{\Phi^+}^{\otimes c}$,
    where $\ket{\Phi^+}=\frac{\ket{00}+\ket{11}}{\sqrt{2}}$ is the maximally-entangled state shared between the sender and the receiver.

\begin{lemma}\cite[Theorem 5]{CSSZ09}
\label{lem:CWS_st_rel}
    Let $\mathcal{Q}$ be an $((n,M,d;c))$ EA-CWS code with   CWS group $S$ and    word operators $W=\{\omega_j\}$. If $W$ forms an abelian group not containing the minus identity, then there exists a stabilizer subgroup of  $\cP^{n+c}$ with $(n-\log_2 M)$ independent generators, whose  stabilized space   is $\mathcal{Q}$.
\end{lemma}
Lemma \ref{lem:CWS_st_rel} allows us to formulate general stabilizer codes or {EA-stabilizer} codes as special cases of EA-CWS codes.

\begin{definition}
    Let $\mathcal{Q}$ be an $((n, M, d; c))$ EA-CWS code with CWS group $S$ and word operators $W=\{\omega_j\}_{j=1}^M$.
    \begin{enumerate}
        \item 
        If $c = 0$, $\mathcal{Q}$ is called a CWS code.
        \item 
        If $W$ is an abelian group not containing the minus identity, $\mathcal{Q}$ is called  an EA stabilizer code.
    
        \item 
        If $\mathcal{Q}$ is an EA-stabilizer code and $c = 0$, then $\mathcal{Q}$ is  a stabilizer code.
    \end{enumerate}
\end{definition}

Consider an $((n,M,d))$ stabilizer code  $\cQ$, stabilized by a stabilizer group $S\subset \cP^n$ with $(n-\log_2 M)$ independent generators by Lemma~\ref{lem:CWS_st_rel}. 
Following the QEC conditions (\ref{error_con_1}) and (\ref{error_con_2}), a Pauli error is  undetectable by $\cQ$ if it commutes with all  elements in $S$.
In addition, if a Pauli error is equivalent to an element  in $S$, up to a phase, it has no effect on $\cQ$. Hence the minimum distance of $\mathcal{Q}$ can be expressed as
    \begin{equation}
    d = {\rm min}\big\{{\rm wt}(\alpha) : \alpha \in \varphi(N(S)) \setminus \varphi(S)\big\}, \label{eq:dmin_stabilizer}
    \end{equation}
where $N(S)$ denotes the normalizer group of $S$ and is defined as
        \[
        N(S) = \{\alpha \in \cP^{n} : \alpha \beta =\beta \alpha, \  \forall \beta \in S\}.
        \]

Similarly,  for an $((n,M,d;c))$ EA-stabilizer code $\cQ$ stabilized by a stabilizer group $S\in\cP^{n+c}$, the minimum distance can be expressed as
  \begin{equation}
    d = {\rm min}\big\{{\rm wt}(\alpha) : \alpha\in\cP^n, \alpha\otimes I^{\otimes c} \in  \varphi(N(S)) \setminus   \varphi(S)\big\}.  \label{eq:dmin_EA}
    \end{equation}
        
 For a stabilizer group $S\in\cP^{n+c}$, the group
\begin{align}
    S_I=\{\alpha\in\cP^n:  \alpha\otimes I^{\otimes c}\in S\}\subset\cP^n \label{eq:isotropic}
\end{align}
is called the \textit{isotropic subgroup} of $S$.

 For an EA stabilizer code $\cQ$ stabilized by a stabilizer group $S$, the  SL weight enumerators $\{A_j\}$ (or  $\{A_{i,j}\}$) and  $\{B_j\}$  (or  $\{B_{i,j}\}$)   correspond to the (split) weight enumerators for the elements in $S$ and in $\varphi(N(S))$, respectively \cite[Theorem 5]{CRSS98}.
 Therefore, the minimum distance of $\cQ$  in (\ref{eq:dmin_EA}) can be characterized using $A_{i,j}$ and $B_{i,j}$, which is in alignment with Lemma~\ref{lem:EA_split_lp}~iii.

For CWS-type codes,  the $B$-type SL weight enumerator has a combinatorial interpretation~\cite[Theorem 1]{NK22}. We extend that to EA-CWS codes.
\begin{lemma}
\label{lem:CWS_exp}
    Let $\mathcal{Q}$ be an $((n,M,d;c))$ EA-CWS code with CWS group $S$ and a set of word operators $W$. 
     Let $S_{I}$ be the isotropic group of $S$. Then, the SL weight enumerator $\{B_{j}\}$ or $\{B_{i,j}\}$ can be regarded as the distance enumerators for  the set $(W\otimes I^{\otimes c})S$ and
    \begin{align}
       \big\lvert (W\otimes I^{\otimes c})S \big\rvert= \sum_{j} B_{j} = \sum_{i,j} B_{i,j} . \label{eq:WS_size}
    \end{align}
    In addition, 
    $\{B_{i, 0}\}_{i=0}^n$ is the distance enumerator for $WS_{I}$.
\end{lemma}
In addition to the linear constraints in Lemma~\ref{lem:EA_split_lp},  Lemma~\ref{lem:CWS_exp} introduces an additional linear constraint for an EA-CWS code.

\begin{proof}
The first part follows the proof in~\cite[Theorem 1]{NK22}. 

Now we consider the second part. Let $|S \rangle$ be the state stabilized by $S$. By the definition of the split weight enumerator $\{B_{i, j}\}$, we have
\begin{align*}
    B_{i, 0} =& \frac{1}{M} \sum_{ \alpha \in \cE_{i}^{n}}{\rm Tr}((\alpha \otimes I^{\otimes c}) \PQ (\alpha \otimes I^{\otimes c}) \PQ)\\
    =& \frac{1}{M} \sum_{ \alpha \in \cE_{i}^{n}}\sum_{k, l} \langle \omega_{k} | \alpha \otimes I^{\otimes c} | \omega_{l} \rangle \langle \omega_{l} | \alpha \otimes I^{\otimes c} |\omega_{k} \rangle\\
    =& \frac{1}{M} \sum_{ \alpha \in \cE_{i}^{n}}\sum_{k, l} \lvert \langle \omega_{k} | \alpha  \otimes I^{\otimes c} |\omega_{l} \rangle\lvert \rvert^{2}\\
    =& \frac{1}{M} \sum_{ \alpha \in \cE_{i}^{n}}\sum_{k, l} \lvert \langle S | (\omega_{k}^{\dagger}\alpha \omega_{l}) \otimes I^{\otimes c}|S \rangle \rvert^{2}.
\end{align*}
One can see that
    \begin{equation*}
    \lvert \langle S | (\omega_{k}^{\dagger}\alpha \omega_{l}) \otimes I^{\otimes c}|S \rangle \rvert = \left\{
    \begin{aligned}
        &1, & &\text{if } \alpha \in \omega_{i}\omega_{j} S_{I};\\
        &0, & &\text{otherwise}.
    \end{aligned}
    \right.
\end{equation*}
This implies that $\{B_{i, 0}\}$ is the distance enumerator for $WS_{I}$.
\end{proof}
By Lemma \ref{lem:EA_split_lp}, the linear programming constraints on the split weight enumerators now can be applied in the context of EA-CWS codes. Moreover, let $\{A_{j}\}$ and $\{B_{j}\}$ be the SL weight enumerators for an EA-CWS code. Using the identities $A_{i} = \sum_{j+k = i}A_{j, k}$ and $B_{i} = \sum_{j+k = i}B_{j, k}$ for $i = 0, \dots, n$, the MacWilliams identities and quantum shadow identities can also be applied.
 We summarize all the linear constraints in the following lemma. 

\begin{lemma}(Linear constraints for EA-CWS codes)  \label{lemma:linear_constraints}
Let $\mathcal{Q}$ be an $((n, M, d; c))$ EA-CWS code with CWS group $S$ and a set of word operators $W$. Let   $S_{I}$ be the isotropic group of $S$.   Let $\{A_{j}\}$, $\{B_{j}\}$, $\{A_{i, j}\}$, and $\{B_{i, j}\}$ be the SL and split weight enumerators for $\mathcal{Q}$, respectively. Then, we have
    \begin{enumerate}
        \item $A_{0, 0} = B_{0, 0} = 1$;
        \item $A_{0, j} = 0$ for $j = 1, \dots, c;$
        \item $B_{i, j} \geq A_{i, j}$ for all $i, j$;
        \item $A_{i, 0} = B_{i, 0}$ if $i = 0, \dots, d-1$ and $A_{d, 0} < B_{d, 0}$;
        \item 
        $B_{i, j} = \frac{M}{2^{n+c}} \sum_{u = 0}^{n} \sum_{v = 0}^{c} K_{i}(u, n) K_{j}(v, c) A_{u, v} \text{ for all } i, j$;
        \item
        $\sum_{i} B_{i, 0} = \lvert WS_{I} \rvert, \text{ and } \sum_{i, j}B_{i, j} = \lvert (W\otimes I^{\otimes n}) S \rvert$;
        \item 
        $\sum_{i+j=k}B_{i+j} = \frac{M}{2^{n}} \sum_{i}\sum_{m+l = i} K_{k}(i, n) A_{m, l} \text{ for } k = 0, \dots, n$;
        \item
        $\sum_{i}(-1)^{i}K_{j}(i, n)A_{i} \geq 0 \text{ for } j = 0, \dots, n$;
        \item
        $\sum_{u = 0}^{n} \sum_{v = 0}^{c} (-1)^{u+v} K_{i}(u, n) K_{j}(v, c)A_{u,v} \geq 0$ for $i = 0, 
        \dots, n$ and $j = 0, \dots, c$.
    \end{enumerate}
\end{lemma}

\section{On the minimum distance of  quantum codes}
\label{sec:code_structure}
 
In~\cite{CSSZ09}, it is shown that a CWS code can be described by a stabilizer group 
and a classical code, and Pauli errors can be transferred to classical errors.
Then, the minimum distance of the CWS code is determined from its associated classical code over the effective classical errors.   For small codes, one can use computer search to determine the distance accordingly. However,   this approach is not scalable for larger codes.
{In this section, we revisit the quantum error correction conditions  and we will provide an upper bound on the minimum distance of an EA-CWS code that can be calculated from associated weight enumerators. }

\begin{definition}
   Let $\mathcal{Q}$ be an $((n, M, d; c))$ EA quantum code and $\PQ$ be the projection operator onto $\mathcal{Q}$. 
   The \textit{generalized stabilizer (GS) sets} for $\cQ$ are defined as a sequence of sets $\{S_{j}^{n, c} : j\in\mathbb{R},  -1\leq j\leq 1\}$, where
\begin{equation}
    S_{j}^{n, c} = \{\beta \in \cE^{n} : \PQ (\beta \otimes I^{\otimes c}) \PQ =  j \PQ\}. \label{eq:stabilizer_sets}
\end{equation} 
The \textit{spectrum} of $\mathcal{Q}$ is defined as
\begin{equation}
    {\rm Sp}(\mathcal{Q}) = \{j : S_{j}^{n, c} \neq \emptyset\}.
\end{equation}
The GS group  for $\mathcal{Q}$ is defined  as 
\begin{align}
     S_g= -S_{-1}^{n, c} \cup S_{+1}^{n, c}. \label{eq:GS_group}
\end{align}
\end{definition}

If $S_{j}^{n, c} \neq \emptyset$ for some $j\in\mathbb{C}$, then $\mathcal{Q}$ is a join eigenspace of all elements  in $S_{j}^{n, c}$ with eigenvalues $j$, which implies that all elements in $S_{j}^{n, c}$ commute. Thus,  $\{S_{j}^{n, c}\}$ can be regarded as a generalized version of the stabilizer group for a stabilizer code.   
Since $\sigma_x, \sigma_y,$ and $\sigma_z$ are Hermitian, possible eigenvalues $j$ are real by ~(\ref{eq:stabilizer_sets}).
The error operators within a GS set have identical effects on the code space, leading to quantum degeneracy.

Notice that $-S_{-1}^{n, c}$ consists of operators with eigenvalues $+1$ on $\mathcal{Q}$. Therefore, we define $S_g$ as the union of $-S_{-1}^{n, c}$ and $S_{+1}^{n, c}$, forming an abelian subgroup of $\cP^n$. Moreover, $S_g$ contains  only Pauli operators that preserve the codespace.

         For $\beta \in S_{0}^{n, c}$,  we have $\langle v | \beta | w \rangle = 0$  for $|v \rangle, |w \rangle \in \mathcal{Q}$. Thus, $\beta\ket{v}=0$ for all $\ket{v}\in\cQ$  and $\cQ$ lies in the null space of all elements in $S_{0}^{n, c}$.
     By QEC conditions~(\ref{error_con_1}) and (\ref{error_con_2}), $S_{0}^{n, c}$ is a set of detectable errors by $\cQ$.
     Note that $S_{0}^{n, c}$ is not necessarily nontrivial.  For example, let  $\mathcal{Q} = \mathbb{C}^{2^{n}}$ with $c=0$ and we have $S_{0}^{n, 0} = \emptyset$.

\begin{example} \label{ex:EPR}
 The EPR pair $\ket{\Phi^+}$ is stabilized by $S=\{$$I\otimes I$, $\sigma_x\otimes \sigma_x$, $-\sigma_y\otimes \sigma_y$, and $\sigma_z\otimes \sigma_z$$\}$.
 By Definition~\ref{eq:stabilizer_sets},     its  spectrum is $\{0,\pm 1\}$ 
 and the corresponding GS sets are 
$S_{1}^{n, c}= \{ I\otimes I, \sigma_x\otimes \sigma_x,   \sigma_z\otimes \sigma_z\}$ and   $S_{-1}^{n, c}= \{  \sigma_y\otimes \sigma_y\}$. Immediately we have $S_g= -S_{-1}^{n, c} \cup S_{+1}^{n, c}=S$.

\end{example}

In the following we show that the minimum distance of a quantum code can be determined from its GS sets.
\begin{lemma}
\label{lem:dist}
Let $\mathcal{Q}$ be an $((n, M, d; c))$ entanglement-assisted quantum code with GS sets $\{S_{j}^{n, c} \}$. Then
\begin{align}
\label{eq:gen_dis}
    d =& {\rm min}\left\{{\rm wt}(\alpha) : \alpha \in \cE^{n} \setminus   \bigcup_{j} S_{j}^{n, c}\right\}.
\end{align}
\end{lemma}

\begin{proof}
We claim that for $\alpha \in \cE^{n}$, $\alpha$ is detectable if and only if $\alpha \in S_{j}^{n, c}$ for some $j$.

Suppose that $\alpha \in S_{j}^{n, c}$ for some $j$.
For any   two basis states $|v \rangle, | w \rangle \in \mathcal{Q}$   and $\langle v | w \rangle=0$,
we have $\langle v | \alpha \otimes I^{\otimes c}| v \rangle = j$  and $\langle v | \alpha \otimes I^{\otimes c} | w \rangle = j \langle v | w \rangle = 0$.
Thus, $\alpha$ is detectable by the QEC conditions (\ref{error_con_1}) and (\ref{error_con_2}).

On the other hand, suppose that $\beta \in {\cE}^{n}$ such that $\beta \otimes I^{\otimes c}$ satisfy the conditions (\ref{error_con_1}) and (\ref{error_con_2}). For $|v \rangle \in \mathcal{Q}$, $\beta \otimes I^{\otimes c} |v \rangle$ can be expressed as $\beta \otimes I^{\otimes c} |v \rangle = i |v \rangle + k |u \rangle$ for some $i, k \in \mathbb{C}$ and $|u \rangle \in (\mathbb{C}^{2})^{\otimes (n+c)}$ such that $\langle v|u\rangle = 0$. Condition (\ref{error_con_1}) implies that the value of $i$ is independent of  $|v \rangle$. Condition (\ref{error_con_2}) implies that  $P|u \rangle =0$,
where $P$ is the projector onto $\cQ$. Thus, we  have $\beta \in S_{i}^{n,c}$.
\end{proof}

Lemma~\ref{lem:dist}  shows that the errors that are detectable by $\cQ$ are contained in the GS sets
and it offers a formula for determining the minimum distance of a quantum code, similar to ~(\ref{eq:dmin_stabilizer}) and (\ref{eq:dmin_EA}). 
Nevertheless, obtaining the full GS sets may remain challenging. A quantum code may have a large spectrum, as demonstrated in the following example.

\begin{example}
 The $((7, 2, 3; 0))$ quantum code constructed in \cite{KT23} has two logical basis states
     \begin{align*}
        &|\bar{0} \rangle = \frac{1}{8} \left(\sqrt{15} | D_{0}^{7}\rangle + \sqrt{7} |D_{2}^{7}\rangle+\sqrt{21} |D_{4}^{7}\rangle -\sqrt{21} |D_{6}^{7}\rangle \right),\\
        &|\bar{1} \rangle = \frac{1}{8} \left(-\sqrt{21} |D_{1}^{7}\rangle + \sqrt{21} |D_{3}^{7}\rangle + \sqrt{7} |D_{5}^{7}\rangle + \sqrt{15} |D_{7}^{7}\rangle\right),
    \end{align*}
    where   \begin{equation*}
        |D_{i}^{j} \rangle = \binom{j}{i}^{-\frac{1}{2}} \sum_{x \in \{0, 1\}^{j}, {\rm wt}(x) = i} |x \rangle
    \end{equation*}
are the Dicke States and ${\rm wt}(x)$ is the number of 1s in a binary string $x$.  
By calculating (\ref{eq:stabilizer_sets}) for all $\beta\in\cE^n$,   the spectrum of this code  can be found as follows:
    \begin{align*}
          &\left\{0, 1, \frac{1}{3}, \frac{1}{5}, \frac{1+\sqrt{5}}{10}, \frac{1-\sqrt{5}}{10}, \frac{1}{15}, \frac{-2}{15}\right\}.
    \end{align*}
Note that this code is not a CWS code and its GS group is trivial.

\end{example}

The following result shows that the spectrum of an EA-CWS code  contains at most three elements $0$ and $\pm 1$, as been shown in Example~\ref{ex:EPR}.

\begin{lemma}
\label{lem:spc_CWS}
    Let $\mathcal{Q}$ be an $((n, M, d; c))$ EA-CWS code.  Then, its spectrum ${\rm Sp}(\mathcal{Q})$  is a subset of $\{0, \pm 1\}$.
\end{lemma}
\begin{proof}

    Let   $|S\rangle\in\mathbb{C}^{2^{n+c}}$ be a stabilizer state, stabilized by a CWS group $S \subset \cP^{n+c}$.
    Since $S$ is maximal, we have $N(S)=\{\pm 1,\pm \bm{i} \}\times S$. 
    Consider $\alpha \in {\cE}^{n}$.
If $\alpha\otimes I^{\otimes c} \in N(S)$, we have $\alpha\otimes I^{\otimes c} \ket{S}=\pm 
\ket{S}$ and thus $ \langle S| (\alpha\otimes I^{\otimes c}) | S \rangle = \pm 1$.

    If $\alpha \notin N(S)$, there exists $\beta \in S$ such that $\beta \alpha = -\alpha \beta$. Then, 
    \begin{align*}
        \langle S| (\alpha\otimes I^{\otimes c}) | S \rangle =& \langle S| (\alpha\otimes I^{\otimes c}) (\beta\otimes I^{\otimes c}) | S \rangle\\
        =& -\langle S| (\beta\otimes I^{\otimes c})(\alpha\otimes I^{\otimes c}) | S \rangle \\
        =& -\langle S|(\alpha\otimes I^{\otimes c}) | S \rangle,
    \end{align*}
    which implies $\langle S| (\alpha\otimes I^{\otimes c}) | S \rangle = 0$. 
    To sum up,  we have 
    \begin{equation*}  \label{lem:cws_1}
        \langle S| (\alpha\otimes I^{\otimes c}) | S \rangle = \left\{
        \begin{aligned}
            &\pm 1, & &\text{if } \alpha \in \pm S;\\
            &0, & &\text{otherwise.}
        \end{aligned}
        \right.
    \end{equation*}

\end{proof}

Based on Lemma~\ref{lem:dist} and Lemma~\ref{lem:spc_CWS}, we can determine the minimum distance of a CWS code using its GS set $S_0^{n,c}$ and the GS group $S_g$.

\begin{lemma}
\label{lem:SI_property}
    Let $\mathcal{Q}$ be an $((n, M, d; c))$ EA-stabilizer code with stabilizer group $S$.
    Then, its GS group is the isotropic subgroup of $S$.

\end{lemma}

\begin{proof}
   Consider $\alpha \in \cP^{n}$. Then, $\langle v | \alpha \otimes I^{\otimes c} | v \rangle = 1$ for all $|v \rangle \in \mathcal{Q}$ if and only if $(\alpha \otimes I^{\otimes c}) | v \rangle = | v \rangle$ for all $|v \rangle \in \mathcal{Q}$. That is, $\alpha \otimes I^{\otimes c} \in S$.
\end{proof}

\begin{remark}
 If $\mathcal{Q}$ is a stabilizer code  with stabilizer group $S$ and GS group $S_g$,
  then we have $S_{g} = S$ by Lemma \ref{lem:SI_property}.
\end{remark}

In Lemma~\ref{lem:SI_property}, we demonstrated that the GS group of an EA stabilizer code is equivalent to the isotropic subgroup. In the subsequent lemma, we establish an upper bound on the minimum distance of an EA-CWS code, using the isotropic subgroup of its CWS group.

 By ~\ref{eq:gen_dis} and Lemma~\ref{lem:spc_CWS}, we have
    \begin{align*}
    d =& {\rm min}\{{\rm wt}(\alpha) : \alpha \in {\cE}^{n} \setminus \varphi (S_{0}^{n, c} \cup S_{g})\}.
    \end{align*}
To establish an upper bound on the minimum distance of $\cQ$, it suffices to identify a nontrivial set of errors that cannot be detected by $\cQ$.
\begin{lemma}
\label{lem:EA-CWS_dist2}
    Let $\mathcal{Q}$ be an $((n, M, d; c))$ EA-CWS code with CWS group $S$ and   word operators $W$.
    Then, we have
    \begin{equation}
        d \leq {\rm min}\{{\rm wt}(\alpha) : \alpha \in  \varphi({W}S_{I}) \setminus \varphi(S_{g})\}, \label{eq:dmin_CWS_upper_1}
    \end{equation}
    where  $S_I$ is the isotropic subgroup of $S$ and $S_g$ is the GS group for $\cQ$.
\end{lemma}

\begin{proof}

    Consider $\alpha \in \ {W}S_{I} \setminus S_{I}$. Let $\alpha = \omega_{i} \beta $ for some $\omega_{i} \in {W}\setminus \{I^{\otimes n}\}$ and $\beta \in S_I$.  
    Let $\ket{v}= (\omega_i\otimes I^{\otimes c})\ket{S}\in\cQ$. 
    We have
    \begin{equation*}
        \langle v |  (\alpha \otimes I^{\otimes c}) | S \rangle = \langle S | (\omega_{i}^\dag \omega_{i} \beta) \otimes I^{\otimes c} | S \rangle = 1,
    \end{equation*}
    which violate the QEC condition~(\ref{error_con_2}).  Thus, $\alpha$ is not detectable by $\cQ$.

On the other hand, consider $\beta\in S_{I} \setminus \varphi (S_{g})$. Since $\beta\notin S_g$, there exists $\ket{v}\in\cQ$ and $\ket{v}\neq \ket{S}$ such that $\bra{v}\beta\otimes I^{\otimes c}\ket{v} \neq \pm 1$.
 However, $\bra{S} \beta\otimes I^{\otimes c} \ket{S}=1 \neq \bra{v}\beta\otimes I^{\otimes c}\ket{v}$.
 By the QEC condition~(\ref{error_con_1}), $\beta\otimes I^{\otimes c}$ is not detectable by $\cQ$.

 Therefore, $W_{I} \setminus S_{g}$ is a set of errors that are undetectable by $\cQ$.

\end{proof}

Lemma \ref{lem:EA-CWS_dist2} provides an upper bound on the minimum distance for an EA-CWS code. For the case of stabilizer codes, ~(\ref{eq:dmin_CWS_upper_1}) holds with equality.

\section{Semidefinite constraints for entanglement-assisted quantum codes}
\label{sec:SDP}

In this section, we extend the SDP method to entanglement-assisted quantum codes. 
 We can use weight enumerators to establish an upper bound on the minimum distance of a code $\cQ$
as shown by the two SL weight enumerators in Lemma~\ref{lemma:SL_enum}.
However, the SL weight enumerators for general quantum codes do not have  combinatorial interpretations. Hence, we specifically consider EA-CWS codes, which can be defined by a CWS group and a set of word operators. Furthermore, the minimum distance of an EA-CWS code can be {upper bounded} using weight enumerators associated to the quantum code by Lemma\ref{lem:EA-CWS_dist2}. 
  
{Let $T_1\subset T_2\subset\cE^n$ be two subsets such that   $T_2\setminus T_1$ is a set of errors that are undetectable by a quantum code $\cQ$. Let $\cW_i(T)$ be the number of operators in a set $T\subset \cE^n$ of weight $i$. Then,  the minimum distance of $\cQ$ is bounded by $\min\{i: \cW_i(T_2)> \cW_i(T_1) \}$.
Consequently, we can derive semidefinite constraints on these sets $T_1$ and $T_2$ as in \cite{Sch05,GST06}  leading to the construction of SDP programs that provide SDP bounds on the minimum distance or size of quantum CWS-type codes.

To derive effective upper bounds, the selection of nontrivial subsets $T_1$ and $T_2$ is a crucial step. For stabilizer codes defined by a stabilizer group $S\subset \cP^n$, the stabilizer group $S$ and its normalizer group $N(S)$ naturally serve as the two sets $T_1$ and $T_2$. This approach is similarly applied to EA stabilizer codes as well~\cite{LA18}.}

For an EA-CWS code $\cQ$ defined by CWS group $S$ and word operators $W$, one may consider the sets $WS_I$ and $S_g$, where $S_g$ is the GS group of $\cQ$ and $S_I$ is the isotropic subgroup of $S$, as presented in~Lemma~\ref{lem:EA-CWS_dist2}. However, establishing the relationships between $S_g$ and $WS_I$ remains a challenge. Alternatively, we  may consider a looser upper bound by 
\begin{equation}
        d \leq {\rm min}\{{\rm wt}(\alpha) : \alpha \in  {W}S_{I} \setminus S_{I}\}, \label{eq:dmin_CWS_upper}
   \end{equation}
where we have used the relation  $S_g\subset S_I$ since $S_g\otimes I^{\otimes c}$ stabilizes the codeword $\ket{S}$.

{Now we establish semidefinite constraints for EA-CWS codes. Our development follows a similar path to that in Section~\ref{sec:SDP_classical}.}
We will examine the group action of 
$$\cG = \left(\prod_{i = 1}^{n} {\mathcal S}_{3}\right) \times {\mathcal S}_{n}$$ on $\cE^{n}$.
The symmetric group ${\mathcal S}_{3}$ will be regarded as permutations on the non-identity Pauli matrices
 $\{\sigma_{x}, \sigma_{y}, \sigma_{z}\}$. 
For  $\beta = b_{1} \otimes \cdots \otimes b_{n} \in \cE^{n}$ and $g= (\pi_{1}, \dots, \pi_{n+1}) \in \cG$, where $a_{i} \in {\mathcal S}_{3}$ for $i < n+1$ and $a_{n+1} \in {\mathcal S}_{n}$,  the action of $g$ on $\beta$ is defined as follows:
\begin{equation*}
    g(\beta) = \pi_{1}(b_{\pi_{n+1}(1)}) \otimes \cdots \otimes \pi_{n}(b_{\pi_{n+1}(n)}).
\end{equation*}
Similarly, we consider the action of $\cG$  on $\cE^{n} \times \cE^{n}$. 
For $g \in \cG$ and $\beta_{1}, \beta_{2} \in \cE^{n}$,  $g$ acts on   $(\beta_{1}, \beta_{2})$ by
     $g(\beta_{1}, \beta_{2}) = (g(\beta_{1}), g(\beta_{2})).$

Observe that $(\alpha_{1}, \beta_{1})$ and $(\alpha_{2}, \beta_{2}) \in \cE^{n} \times \cE^{n}$ have the same orbit  under the action of $\cG^{n}$ if and only if 
\begin{align}
    &{\rm wt}(\alpha_{1}) = {\rm wt}(\alpha_{2}),\\ 
    &{\rm wt}(\beta_{1}) = {\rm wt}(\beta_{2}),\\
    &\lvert {\rm Supp}(\alpha_{1}) \cap {\rm Supp}(\beta_{1}) \rvert = \lvert {\rm Supp}(\alpha_{2}) \cap {\rm Supp}(\beta_{2}) \rvert,\\
    &{\rm wt}(\alpha_{1}\beta_{1}) = {\rm wt}(\alpha_{2}\beta_{2}).  
\end{align}
Now, we define   $4^{n} \times 4^{n}$ matrices $M_{i, j}^{t, p}$ for $1\leq i, j\leq n$, $0 \leq p \leq t \leq {\rm min}\{i, j\}$ and $i+j \leq n+t$  to characterize the pairwise relations between any two operators in $\cE^n$.  $M_{i, j}^{t, p}$ is indexed by the elements of $\cE^n$ 
and its $(\alpha,\beta)$-th entry  for $\alpha,\beta\in \cE^n$ is
\begin{equation}
    \left(M_{i, j}^{t, p}\right)_{\alpha, \beta} = \left\{
    \begin{aligned}
        &1, & &\text{if } {\rm wt}\left(\alpha\right) = i, {\rm wt}\left(\beta\right) = j,\\
        & & &\lvert {\rm Supp}(\alpha) \cap {\rm Supp}(\beta) \rvert = t, \\
        & & &\lvert {\rm Supp}(\alpha) \setminus {\rm Supp}(\alpha\beta) \rvert = p;\\
        &0, & &\text{otherwise}.
    \end{aligned}
    \right.
\end{equation}
Let ${\mathcal A}_{n, 4}$ be the algebra generated by the matrices $M_{i, j}^{t, p}$. Then, ${\mathcal A}_{n, 4}$ is a $\mathbb{C}^{*}$-algebra, called the Terwilliger algebra of the Hamming scheme ${\mathcal H}(n, 4)$ \cite{GST06}. For details on the Terwilliger algebra, we recommend readers to refer to \cite{Ter92} and \cite{Ter93}.

As discussed at the beginning of this section, the two sets, $T_1$ and $T_2$, satisfy the condition that $T_2\setminus T_1$ consists of errors undetectable by $\cQ$ and can be utilized to establish an upper bound on the minimum distance of $\cQ$. As demonstrated in Section~\ref{sec:SDP_classical}, positive semidefinite matrices are defined for classical codes. Now, we will define positive semidefinite matrices with respect to the two sets $T_1$ and $T_2$.

Let $\alpha \in T_{1}$ and $\beta \in T_{2}$.  We define semidefinite matrices $R_{T_{1}}^{\alpha}$, $R_{T_{1}, T_{2}}^{\alpha}$, $R_{T_{2}, T_{1}}^{\beta }$ and $R_{T_{2}}^{\beta }$, indexed by the elements in $\cE^{n}$, as follows:

\begin{align*}
    &R_{T_{1}}^{\alpha} = \sum_{g \in \cG} \chi^{g(\varphi(\alpha T_{1}))} \left(\chi^{g(\varphi(\alpha T_{1}))}\right)^{T},\\
    &R_{T_{1}, T_{2}}^{\alpha } = \sum_{g \in \cG} \chi^{g(\varphi(\alpha T_{2}))} \left(\chi^{g(\varphi(\alpha T_{2}))}\right)^{T},\\
    &R_{T_{2}, T_{1}}^{\beta } = \sum_{g \in \cG} \chi^{g(\varphi(\beta  T_{1}))} \left(\chi^{g(\varphi(\beta T_{1}))}\right)^{T},\\
    &R_{T_{2}}^{\beta} = \sum_{g \in \cG} \chi^{g(\varphi(\beta T_{2}))} \left(\chi^{g(\varphi(\beta T_{2}))}\right)^{T},
\end{align*}
where  $\chi^{\cD}$, for $\cD \subset \cE^{n}$, is a column vector indexed by the elements in $\cE^{n}$ and is defined by 
\begin{equation*}
    \left(\chi^{\cD}\right)_{\alpha} = \left\{
    \begin{aligned}
        &1, & & \alpha \in \cD;\\
        &0, & &\text{otherwise.}
    \end{aligned} \right.
\end{equation*}
  Then, we define the following matrices, all of which are positive semidefinite since they are convex combinations of positive semidefinite matrices.
\begin{align}
    &R_{T_{1}} = \left\lvert T_{1} \right\rvert^{-1} \sum_{\alpha \in T_{1}} R_{T_{1}}^{\alpha},\\  &R_{T_{1}}' = \left\lvert 4^{n} - \lvert T_{1} \rvert \right\rvert^{-1} \sum_{\alpha \in \cE^{n} \setminus T_{1}} R_{T_{1}}^{\alpha}, \label{eq:RT10}\\
    &R_{T_{1}, T_{2}} = \left\lvert T_{1} \right\rvert^{-1} \sum_{\alpha \in T_{1}} R_{T_{1}, T_{2}}^{\alpha}, \\
    &R_{T_{1}, T_{2}}' = \left\lvert 4^{n} - \lvert T_{1} \rvert \right\rvert^{-1} \sum_{a \in \cE^{n} \setminus T_{1}} R_{T_{1}, T_{2}}^{\alpha},\\
    &R_{T_{2}, T_{1}} = \left\lvert T_{2} \right\rvert^{-1} \sum_{\beta \in T_{2}} R_{T_{2}, T_{1}}^{\beta},\\
    &R'_{T_{2}, T_{1}}=  \left\lvert 4^{n} - \lvert T_{2} \rvert \right\rvert^{-1} \sum_{\beta \in \cE^{n} \setminus T_{2}} R_{T_{2}, T_{1}}^{\beta},\\
    &R_{T_{2}} = \left\lvert T_{2} \right\rvert^{-1} \sum_{\beta \in T_{2}} R_{T_{2}}^{\beta},\\ 
    & R'_{T_{2}} =  \left\lvert 4^{n} - \lvert T_{2} \rvert \right\rvert^{-1} \sum_{\beta \in \cE^{n} \setminus T_{2}} R_{T_{2}}^{\beta}.
\end{align}
For two sets $T_1$ and $T_2\subset \cE^n$, we define 
 \begin{equation*}
\begin{aligned}
        \sigma_{i, j}^{t, p} = \lvert &\{(\alpha, \beta, \gamma) \in T_{1} \times T_{1} \times T_{1} : {\rm wt}(\alpha\beta) = i, {\rm wt}(\alpha\gamma) = j,\\
        &\lvert {\rm Supp}(\alpha\beta) \cap {\rm Supp}(\alpha\gamma) \rvert = t, \\
        &
        \lvert {\rm Supp}(\alpha\beta) \setminus {\rm Supp}(\beta\gamma) \rvert = p\} \rvert;
\end{aligned}
\end{equation*}
\begin{equation*}
\begin{aligned}
        \mu_{i, j}^{t, p} = \lvert &\{(\alpha, \beta, \gamma) \in T_{1} \times T_{2} \times T_{2} : {\rm wt}(\alpha\beta) = i, {\rm wt}(\alpha\gamma) = j,\\
        &\lvert {\rm Supp}(\alpha\beta) \cap {\rm Supp}(\alpha\gamma) \rvert = t, \\
        &
        \lvert {\rm Supp}(\alpha\beta) \setminus {\rm Supp}(\beta\gamma) \rvert = p\} \rvert;
\end{aligned}
\end{equation*}
\begin{equation*}
\begin{aligned}
        \nu_{i, j}^{t, p} = \lvert &\{(\alpha, \beta, \gamma) \in T_{2} \times T_{1} \times T_{1} : {\rm wt}(\alpha\beta) = i, {\rm wt}(\alpha\gamma) = j,\\
        &\lvert {\rm Supp}(\alpha\beta) \cap {\rm Supp}(\alpha\gamma) \rvert = t, \\
        &
        \lvert {\rm Supp}(\alpha\beta) \setminus {\rm Supp}(\beta\gamma) \rvert = p\} \rvert;
\end{aligned}
\end{equation*}
\begin{equation*}
\begin{aligned}
        \eta_{i, j}^{t, p} = \lvert &\{(\alpha, \beta, \gamma) \in T_{2} \times T_{2} \times T_{2} : {\rm wt}(\alpha\beta) = i, {\rm wt}(\alpha\gamma) = j,\\
        &\lvert {\rm Supp}(\alpha\beta) \cap {\rm Supp}(\alpha\gamma) \rvert = t, \\
        &
        \lvert {\rm Supp}(\alpha\beta) \setminus {\rm Supp}(\beta\gamma) \rvert = p\} \rvert.
\end{aligned}
\end{equation*}
Each of these four variables represents the size of possible combinations of $\alpha, \beta, \gamma$ belonging to $T_1$ or $T_2$, satisfying the specified distance relations.
In addition, we have the following identities.
\begin{lemma}
\label{lem:sdp_cons_gen}
\begin{align*}
    &R_{T_{1}} = \sum_{i, j, t, p} x_{i, j}^{t, p} M_{i, j}^{t, p};\\
    &R_{T_{1}}' = \frac{\lvert T_{1} \rvert}{4^{n} - \lvert T_{1} \rvert}\sum_{i, j, t, p} (x_{i+j-t-p, 0}^{0, 0} - x_{i, j}^{t, p}) M_{i, j}^{t, p};\\
    &R_{T_{1}, T_{2}} = \sum_{i, j, t, p} u_{i, j}^{t, p} M_{i, j}^{t, p};\\
    &R_{T_{1}, T_{2}}' = \frac{\lvert T_{1} \rvert}{4^{n} - \lvert T_{1} \rvert}\sum_{i, j, t, p} (u_{i+j-t-p, 0}^{0, 0}- u_{i, j}^{t, p}) M_{i, j}^{t, p};\\
    &R_{T_{2}, T_{1}} = \sum_{i, j, t, p} v_{i, j}^{t, p} M_{i, j}^{t, p};\\
    &R_{T_{2}, T_{1}}' = \frac{\lvert T_{1} \rvert}{4^{n} - \lvert T_{2} \rvert}\sum_{i, j, t, p} (v_{i+j-t-p, 0}^{0, 0}- v_{i, j}^{t, p}) M_{i, j}^{t, p};\\
    &R_{T_{2}} = \sum_{i, j, t, p} y_{i, j}^{t, p} M_{i, j}^{t, p};\\
    &R_{T_{2}}' = \frac{\lvert T_{2} \rvert}{4^{n} - \lvert T_{2} \rvert}\sum_{i, j, t, p} (y_{i+j-t-p, 0}^{0, 0}- y_{i, j}^{t, p}) M_{i, j}^{t, p};
\end{align*}
where
\begin{align*}
    &x_{i, j}^{t, p} = \frac{1}{\lvert T_{1} \rvert \gamma_{i, j}^{t, p}} \sigma_{i, j}^{t, p}, \quad u_{i, j}^{t, p} = \frac{1}{\lvert T_{1} \rvert \gamma_{i, j}^{t, p}} \mu_{i, j}^{t, p},\\
    &v_{i, j}^{t, p} = \frac{1}{\lvert T_{1} \rvert \gamma_{i, j}^{t, p}} \nu_{i, j}^{t, p},
    \quad y_{i, j}^{t, p} = \frac{1}{\lvert T_{2} \rvert \gamma_{i, j}^{t, p}} \eta_{i, j}^{t, p},
\end{align*}
and
\begin{equation*}
    \gamma_{i, j}^{t, p} = 2^{t-p}3^{i+j-t}\binom{n}{p,t-p,i-t,j-t}.
\end{equation*}
\end{lemma}

\begin{proof}
The proof is similar to the proof of Theorem~\ref{thm:sdp_matrix}, as demonstrated in \cite{GST06}. For the sake of completeness, we will prove the first case, and the other cases follow similarly.

Since $R_{T_{1}}$ and $R'_{T{1}}$ are indexed by the elements in $\cE$, we consider the action of $\cG$  on $R_{T_{1}}$ and $R'_{T{1}}$ by permuting the row and column indices. Both $R_{T_{1}}$ and $R'_{T{1}}$ are invariant under the action of $\mathcal{G}$, indicating that $R_{T_{1}}, R'_{T{1}} \in \mathcal{A}_{n, 4}$. As a result, they can be expressed as a linear combination of $M_{i, j}^{t, p}$. Let 
\begin{align}R_{T_{1}} = \sum_{i, j, p, t}r_{i, j}^{t, p} M_{i, j}^{t, p}, \label{eq:RT1}
\end{align}
where $r_{i, j}^{t, p}\in\mathbb{C}$.
    
Consider the following positive semidefinite matrix

\begin{align*}
    U =& \lvert T_{1} \rvert R_{T_{1}} + (4^{n} - \lvert T_{1} \rvert) R'_{T_{1}} \\
    =& \sum_{\alpha \in \cE^{n}} R_{T_{1}}^{\alpha}\\ 
    =& \sum_{\alpha \in \cE^{n}} \sum_{g \in \cG} \chi^{g(\varphi(\alpha T_{1}))} \left(\chi^{g(\varphi(\alpha T_{1}))}\right)^{T}.
\end{align*}
One can check that $U$ remains invariant under the action of all permutations of row and column indices, making it an element of the Bose–Mesner algebra $\mathcal{B}_{n, 4}$ associated with the Hamming scheme $\mathcal{H}(n, 4)$   \cite{Del73,GST06}. As a reminder, the algebra $\mathcal{B}_{n, 4}$ is generated by the matrices
    \begin{equation*}
    \tilde{M}_{k} = \sum_{i+j-t-p = k} M_{i, j}^{t, p}
    \end{equation*}
    for $k=0,1,\dots,n$.
    Thus $U$ can be expressed as a linear combinations of $\tilde{M}_{k}$.     Let $U = \sum_{k} z_{k}\tilde{M}_{k}$ for some $z_k\in\mathbb{C}$. 
    For $\alpha \in \cE^{n}$ with ${\rm wt}(\alpha) = k$, one can observe that 
    \begin{equation*}
    z_{k} = \big(U\big)_{\alpha, I^{\otimes n}} = \lvert T_{1} \rvert \big(R_{T_{1}}\big)_{\alpha, I^{\otimes n}} = \lvert T_{1} \rvert r_{k, 0}^{0, 0},
    \end{equation*}
    where we have used the fact that $(R'_{T_{1}})_{\alpha, I^{\otimes n}} = 0$ by definition. Then
    \begin{align*}
        (4^{n} - \lvert T_{1} \rvert) R'_{T_{1}} =& U - \lvert T_{1} \rvert R_{T_{1}}\\
        =& \lvert T_{1} \rvert \left(\sum_{k}r_{k, 0}^{0, 0} \tilde{M}_{k} - \sum_{i, j, p, t} r_{i, j}^{t, p} M_{i,j}^{t, p}\right)\\
        =& \lvert T_{1} \rvert \sum_{i, j, p, t}\left(r_{i+j-p-t, 0}^{0, 0}  - r_{i, j}^{t, p}\right) M_{i,j}^{t, p},
    \end{align*}
    which implies that $R'_{T_{1}} = \frac{\lvert T_{1} \rvert}{4^{n}-\lvert T_{1} \rvert}\sum_{i, j, p, t}(r_{i+j-p-t, 0}^{0, 0}  - r_{i, j}^{t, p}) M_{i,j}^{t, p}$.

    It remains to determine $r_{i, j}^{t, p}$. Recall that $\langle M_{i, j}^{t, p}, M_{i, j}^{t, p}\rangle = \gamma_{i, j}^{t, p}$ \cite{GST06}. Let 
    \begin{align*}
        \cE_{i, j}^{t, p} = &\big\{(\alpha, \beta, \gamma) \in \cE^{n} \times \cE^{n} \times \cE^{n} : {\rm wt}(\alpha\beta) = i, {\rm wt}(\alpha\gamma) = j,\\
        &\lvert {\rm Supp}(\alpha\beta) \cap {\rm Supp}(\alpha\gamma) \rvert = t, \\
        &
        \lvert {\rm Supp}(\alpha\beta) \setminus {\rm Supp}(\beta\gamma) \rvert = p \big\}.
    \end{align*}
    By (\ref{eq:RT10}), we hav
    \begin{align*}
        \langle R_{T_1}, M_{i, j}^{t, p} \rangle =& \frac{1}{\lvert T_{1} \rvert} \sum_{\alpha  \in T_{1}} \langle R_{T_{1}}^{ \alpha }, M_{i, j}^{t, p} \rangle\\
        =& \frac{1}{\lvert T_{1} \rvert} \sum_{\alpha \in T_{1}} \lvert (\{ \alpha  \} \times T_{1} \times T_{1}) \cap \cE_{i, j}^{t, p} \rvert \\
        =& \frac{\sigma_{i, j}^{t, p}}{\lvert T_{1} \rvert}.
    \end{align*}
    Also, by (\ref{eq:RT1}),   $\langle R_{T_1}, M_{i, j}^{t, p} \rangle= r_{i, j}^{t, p}\gamma_{i, j}^{t, p}$.
     Thus $r_{i, j}^{t, p} = \frac{\sigma_{i, j}^{t, p}}{\lvert T_{1} \rvert \gamma_{i, j}^{t, p}} = x_{i, j}^{t, p}$.
\end{proof}

\begin{remark}
Unlike the SDP method for classical codes, we deal with four groups of positive semidefinite matrices in the quantum case, corresponding to the sets $T_{1}$ and $T_{2}$. To effectively characterize the minimum distance, we must establish semidefinite constraints that capture the distance relations between $T_{1}$ and $T_{2}$. For this reason, we require at least two groups of positive semidefinite matrices, as indicated in Lemma \ref{lem:sdp_cons_gen}, to implement the SDP method for quantum codes.

\end{remark}

Consequently, we have the following semidefinite constraints.
\begin{theorem}(Semidefinite constraints for quantum codes)
\label{thm:sdp_cons}
Let $\mathcal{Q}$ be a quantum code of length $n$ and minimum distance $d$. Let $T_{1}$ and $T_{2}$ be two subsets of $\cE^{n}$ satisfying  that $T_{1} \subset T_{2}$ and $T_2\setminus T_1$ is a set of errors undetectable by $\cQ$. Then, the variables $x_{i, j}^{t, p}$, $u_{i, j}^{t, p}$, $v_{i, j}^{t, p}$, and $y_{i, j}^{t, p}$ satisfy the follows constraints.
\begin{enumerate}
    \item
    The matrices $R_{T_{1}}$, $R_{T_{1}, T_{2}}$, $R_{T_{2}, T_{1}}$, $R_{T_{2}}$, $R'_{T_{1}}$, $R'_{T_{1}, T_{2}}$, $R'_{T_{2}, T_{1}}$, and $R'_{T_{2}}$ are positive semidefinite;
    \item
    $\lvert T_{1} \rvert = \sum_{i} \gamma_{i, 0}^{0, 0}x_{i, 0}^{0, 0} = \sum_{i} \gamma_{i, 0}^{0, 0} v_{i, 0}^{0, 0}$;
    \item
    $\lvert T_{2} \rvert = \sum_{i} \gamma_{i, 0}^{0, 0} u_{i, 0}^{0, 0} = \sum_{i} \gamma_{i, 0}^{0, 0} y_{i, 0}^{0, 0}$;
    \item
    $0 \leq  x_{i, j}^{t, p} \leq  u_{i, j}^{t, p},  v_{i, j}^{t, p} \leq \frac{\lvert T_{2} \rvert}{ \lvert T_{1} \rvert} y_{i, j}^{t, p}$, for all possible $i, j, p, t$;
    \item
    If $i < d$ or $j < d$, $x_{i, j}^{t, p} = v_{i, j}^{t, p}$;
    \item
    If $i < d$ and $j < d$, $x_{i, j}^{t, p} = u_{i, j}^{t, p} = v_{i, j}^{t, p}$;
    \item
    $x_{0, 0}^{0, 0} = u_{0, 0}^{0, 0} = y_{0, 0}^{0, 0} = v_{0, 0}^{0, 0} = 1$;
    \item
    $0 \leq x_{i, j}^{t, p} \leq x_{i, 0}^{0, 0}$, $0 \leq u_{i, j}^{t, p} \leq u_{i, 0}^{0, 0}$, 
    $0 \leq v_{i, j}^{t, p} \leq v_{i, 0}^{0, 0}$, 
    $0 \leq y_{i, j}^{t, p} \leq y_{i, 0}^{0, 0}$;
    \item
    $x_{i, j}^{t, p} = x_{i', j'}^{t', p'}$, $u_{i, j}^{t, p} = u_{i', j'}^{t', p'}$, $v_{i, j}^{t, p} = v_{i', j'}^{t', p'}$, and $y_{i, j}^{t, p} = y_{i', j'}^{t', p'}$ if $t-p = t'-p'$ and $(i$, $j$, $i+j-t-p)$ is a permutation of $(i', j', i'+j'-t'-p')$.
\end{enumerate}
\end{theorem}

\begin{proof}

Constraints 1) to 7) are based on the definitions of the variables $x_{i, j}^{t, p}$, $u_{i, j}^{t, p}$, $v_{i, j}^{t, p}$, $y_{i, j}^{t, p}$. Constraints 5) and 6) correspond to the minimum distance requirement. Constraint 8) enforces the positive semidefiniteness of matrices $R_{A}$, $R_{A, B}$, $R_{B, A}$ and $R_{B}$. Constraint 9) deals with the permutation of indices.
\end{proof}

To implement the semidefinite constraints, we can employ the following theorem to block diagonalize the matrices $R_{T_{1}}$, $R_{T_{1}, T_{2}}$, $R_{T_{2}, T_{1}}$, $R_{T_{2}}$, $R'_{T_{1}}$, $R'_{T_{1}, T_{2}}$, $R'_{T_{2}, T_{1}}$, and $R'_{T_{2}}$.

\begin{theorem}\cite[Block diagonal formula]{GST06}
\label{thm:block_dia}
The algebra ${\mathcal A}_{n, 4}$ is isomorphic to 
\begin{equation*}
    \bigoplus_{0 \leq a \leq k \leq n+a-k} \mathbb{C}^{(n+a-2k+1)^{2}}
\end{equation*}
with the isomorphism that maps $\sum_{i, j, p, t}x_{i, j}^{t, p}M_{i, j}^{t, p}$ to
\begin{equation}
    \bigoplus_{0 \leq a \leq k \leq n+a-k}\left(\sum_{t, p}\alpha(i, j, p, t, a, k)x_{i, j}^{t, p}\right)_{i, j = k}^{n+a-k},
\end{equation}
where
\begin{align*}
    \alpha(i, j, p, t, a, k) =& \beta_{i-a, j-a, k-a}^{n-a, t-a}(q-1)^{\frac{1}{2}(i+j)-t}\sum_{g = 0}^{p}(-1)^{a-g}\\
    &\cdot\binom{a}{g} \binom{t-a}{p-g}(q-2)^{t-a-p+g},
\end{align*}
\begin{align*}
    \beta_{i, j, k}^{m, t} =& \sum_{u = 0}^{m}(-1)^{t-u}\binom{u}{t}\binom{m-2k}{m-k-u}\binom{m-k-u}{i-u}\\
    &\cdot \binom{m-k-u}{j-u}.
\end{align*}
\end{theorem}
Therefore, $R_{T_1}$ is positive semidefinite 
if and only if $  \left(\sum_{t, p}\alpha(i, j, p, t, a, k)x_{i, j}^{t, p}\right)_{i, j = k}^{n+a-k}$ is positive semidefinite
for $0 \leq a \leq k \leq n+a-k$.

In the case of classical nonlinear codes, the semidefinite constraints include Delsarte's inequalities~\cite{Sch05}. 
In this setting, we arrive a similar outcome
by using Lemma \ref{lem:sdp_cons_gen}.
\begin{corollary}
\label{coro:gen_lin}
 The variables $x_{k, 0}^{0, 0}$, $u_{k, 0}^{0, 0}$, and $y_{k, 0}^{0, 0}$ satisfy the following linear constraints:
    \begin{align}
        &\sum_{k}3^{k}\binom{n}{k}x_{k, 0}^{0, 0} K_{i}(k, n) \geq 0; \label{equ:lin_1}\\
        &\sum_{k}3^{k}\binom{n}{k}u_{k, 0}^{0, 0} K_{i}(k, n) \geq 0; \label{equ:lin_2}\\
        &\sum_{k}3^{k}\binom{n}{k}y_{k, 0}^{0, 0} K_{i}(k, n) \geq 0, \label{equ:lin_3}
    \end{align}
    for $i=0,1, \dots, n$, where $K_{i}(k, n)$ 
is the $4$-ary Kravchuk polynomial given in ~(\ref{eq:kraw}).
\end{corollary}
\begin{proof}
    We will prove the first case and the other cases follow similarly. Consider the matrix
    \begin{equation*}
        \tilde{R}_{T_{1}} = R_{T_{1}} + \frac{4^{n}- \lvert T_{1} \rvert}{\lvert T_{1} \rvert}R'_{T_{1}},
    \end{equation*}
    which is positive semidefinite since it is a convex combination of positive semidefinite matrices. By Lemma \ref{lem:sdp_cons_gen}, we have
    \begin{equation*}
        \tilde{R}_{T_{1}} = \sum_{i, j, t, p} x_{i+j-t-p, 0}^{0, 0} M_{i, j}^{t, p} = \sum_{k} x_{k, 0}^{0, 0} \left( \sum_{i+j-t-p = k}M_{i, j}^{t, p}\right),
    \end{equation*}
    which implies that $\tilde{R}_{T_{1}}$ is a linear combination of matrices in the Bose-Mesner algebra of the Hamming scheme \cite{GST06}, and $\tilde{R}_{T_{1}}$ is diagonalizable  with eigenvalues
    \begin{equation*}
        \sum_{i}x_{i, 0}^{0, 0} K_{i}(j, n)
    \end{equation*}
    for $j =0,1, \dots, n$ \cite{Del73,Del75}. Applying the symmetry relation of the Kravchuk polynomial \cite{MS77}
    \begin{equation*}
        3^{j} \binom{n}{j} K_{i}(j, n) = 3^{i} \binom{n}{i} K_{j}(i, n),
    \end{equation*}
    we obtain
    \begin{equation*}
        \frac{1}{3^{j}\binom{n}{j}}\sum_{i}3^{i}\binom{n}{i}x_{i, 0}^{0, 0} K_{j}(i, n) \geq 0.
    \end{equation*}
\end{proof}

{
\begin{remark}
    The variables $v_{k, 0}^{0, 0}$ also satisfy
    \[
    \sum_{k}3^{k}\binom{n}{k}v_{k, 0}^{0, 0} K_{i}(k, n) \geq 0,
    \]
    for $i = 0, 1, \dots, n$. However, by the definition, we have $x_{k, 0}^{0, 0} = v_{k, 0}^{0, 0}$ for all $k$. Those linear constraints are the same as  (\ref{equ:lin_1}).
\end{remark} 
}

Corollary \ref{coro:gen_lin} suggests that the semidefinite constraints are more general than the MacWilliams identities. For example, if $3^{j}\binom{n}{k} x_{k, 0}^{0,0} = A_{j}$ for all $j$, then \ref{equ:lin_1} represents the positiveness of the weight enumerator $\{B_{j}\}$.

\section{{Linear and} Semidefinite programming for quantum codes}
\label{sec:app}

In this section, we construct {linear and} semidefinite programs for EA-CWS codes.

Using Lemma~\ref{lemma:linear_constraints}, we have an LP for general EA quantum codes,
which can be used to construct upper bounds on the distance of general EA quantum codes. In particular, the resulting LP bounds for $[[n,k,d;c]]$ EA-stabilizer codes  
improve several entries on the results in ~\cite{Gra_tab}
and are summarized in Table~\ref{table:stabilizer},
where an $[[n,k,d;c]]$ EA-stabilizer code is an $((n,2^k,d;c))$ EA quantum code.
Remarkably, these upper bounds have matched the corresponding lower bounds provided by known codes~\cite{Gra_tab}, indicating the optimality of these codes  for the given parameters $n,k,c$.

\begin{table*}[htbp]
\centering
\begin{threeparttable}
\caption{{Optimal $[[n,k;c]]$ EA-stabilizer codes. }
}
\label{table:stabilizer}
\begin{tabular}{|c|c|c|c|c|c|c|c|c|c|c|c|}
\hline
\diagbox{$n$}{$k$} & $0$ & $1$ & $2 $& $3$ & $4$ & $5$ & $6$ & $7$ & $8$ & $9$ & $10$\\
\hline
$7$ & & & & & $[[7,4,3;3]]$ & & & & & & \\
\hline
$8$ & & & & & & $[[8,5,3;3]]$ & & & & & \\
\hline
$9$ & & & & & $[[9,4,5;5]]$ & 
$[[9,5,3;2]]$, $[[9,5,4;4]]$ & $[[9,6,3;3]]$ & & & & \\
\hline
$10$ & & & & $[[10,3,4;2]]$ & $[[10,4,4;2]]$ & & $[[10,6,3;2]]$ & & & & \\
\hline
\end{tabular}
\end{threeparttable}
\end{table*}

Let $\mathcal{Q}$ be an EA-CWS code with   CWS group $S$ and a set of   word operators $W$. Let $S_{I}$ be the isotropic group of $S$. As discussed in the previous section, we have a upper bound on the minimum distance of $\cQ$ in ~(\ref{eq:dmin_CWS_upper}). Thus we apply Theorem \ref{thm:sdp_cons} with $T_{1} = S_{I}$ and $T_{2} = WS_{I}$ to obtain semidefinite constraints for $\mathcal{Q}$.

Combining the semidefinite constraints with the linear constraints in Lemma~\ref{lemma:linear_constraints}, we have the following theorem.

\begin{theorem}(SDP for EA-CWS codes)
    Let $\mathcal{Q}$ be an $((n, M, d; c))$ EA-CWS code. Consider the variables $x_{i, j}^{t, p}$, $u_{i, j}^{t, p}$, $v_{i, j}^{t, p}$, $y_{i, j}^{t, p}$ for $1 \leq i, j \leq n$, $0 \leq p \leq t \leq {\rm min}\{i, j\}$, and $i+j \leq n+t$, and the variables {$A_{r, s}$, $B_{r, s}$ for $r = 0, \dots, n$ and $s = 0, \dots, c$}. Then, there is a feasible solution to the conditions in Theorem \ref{thm:sdp_cons} and Lemma \ref{lemma:linear_constraints} together with the following equality constraints 
    {
    \begin{align*}
        &3^{i}\binom{n}{i}y_{i, 0}^{0, 0} = \sum_{r+s = i} B_{r, s} \text{ for } i = 0, \dots, n.
    \end{align*}
    }
\end{theorem}

We provide a summary of the upper bounds on the size of EA-CWS codes for various lengths and levels of entanglement using  SDP in Table \ref{table:sdp}. Several of these upper bounds represent improvements over the bounds that rely solely on linear constraints~{from Lemma~\ref{lemma:linear_constraints}}.

In cases where no entanglement assistance is present (i.e., $c = 0$), the semidefinite constraints are shown to outperform the linear constraints from MacWilliams identities in some scenarios. For example, the CWS code with parameters $((10, 9, 4; 0))$ is found to be feasible using MacWilliams identities but infeasible under the semidefinite constraints. However, the same parameter  $((10, 9, 4; 0))$ can be excluded by the combination of MacWilliams identities and the  shadow identities. As of now, all SDP bounds are equal to their corresponding  LP bounds.

We can also note that as the parameter $d$ grows significantly in comparison to the length $n$, the utility of the LP bounds decreases. This is particularly evident for LP bounds with values that are nontrivial powers of $2$. However, in such cases, the SDP bounds continue to provide tight and meaningful results.

{
\begin{table*}[htbp]
\centering
\begin{threeparttable}
\caption{The upper bounds on the size of EA-CWS codes. 
{
The notation $a\ (b)$ indicates that $a$ is the result obtained through SDP, while $b$ is the result obtained through  LP. If an entry contains only one value, it means that the SDP and LP bounds are identical. }
}
\label{table:sdp}
\begin{tabular}{|c|c|c|c|c|c|c|c|c|}
\hline
\diagbox{$(n, c)$}{$d$} & $3$ & $4$ & $5$ & $6$ & $7$ & $8$ & $9$ & $10$\\
\hline
$(3, 1)$ & $2$ & & & & & & & \\
\hline
$(4, 1)$ & $2$ & $2$ & & & & & & \\
\hline
$(5, 1)$ & $4$ & $2$ & $2$ & & & & & \\
\hline
$(6, 1)$ & $5$ & $2$ & $2$ & $2$ & & & & \\
\hline
$(7, 1)$ & $9$ & $2$ & $2$ & $2$ & $2$ & & & \\
\hline
$(8, 1)$ & $18$ & $3$ & $2$ & $2$ & $2$ & $2$ & & \\
\hline
$(9, 1)$ & $32$ & $5$ & $2$ & $2$ & $2$ & $2$ & $2$ & \\
\hline
$(10, 1)$ & $58$ & $13$ & $2$ & $2$ & $2$ & $2$ & $2$ & $2$\\
\hline
$(3, 2)$ & $3\ (4)$ & & & & & & & \\
\hline
$(4, 2)$ & $4$ & $2\ (4)$ & & & & & & \\
\hline
$(5, 2)$ & $8$ & $4$ & $2\ (4)$ & & & & & \\
\hline
$(6, 2)$ & $11$ & $4$ & $4$ & $2\ (4)$ & & & & \\
\hline
$(7, 2)$ & $19$ & $5$ & $4$ & $3\ (4)$ & $2\ (4)$ & & & \\
\hline
$(8, 2)$ & $36$ & $8$ & $4$ & $4$ & $3\ (4)$ & $2\ (4)$ & & \\
\hline
$(9, 2)$ & $67$ & $18$ & $4$ & $4$ & $4$ & $3\ (4)$ & $2\ (4)$ & \\
\hline
$(10, 2)$ & $118$ & $31$ & $7$ & $4$ & $4$ & $4$ & $3\ (4)$ & $2\ (4)$\\
\hline
$(4, 3)$ & $8$ & $3\ (8)$ & & & & & & \\
\hline
$(5, 3)$ & $16$ & $6\ (8)$ & $3\ (8)$ & & & & & \\
\hline
$(6, 3)$ & $22$ & $8$ & $5\ (8)$ & $3\ (8)$ & & & & \\
\hline
$(7, 3)$ & $38$ & $11$ & $8$ & $4\ (8)$ & $2\ (8)$ & & & \\
\hline
$(8, 3)$ & $73$ & $18$ & $8$ & $8$ & $4\ (8)$ & $2\ (8)$ & & \\
\hline
$(9, 3)$ & $140$ & $36$ & $9$ & $8$ & $8$ & $4\ (8)$ & $2\ (8)$ & \\
\hline
$(10, 3)$ & $237$ & $70$ & $16$ & $8$ & $8$ & $7\ (8)$ & $4\ (8)$ & $3\ (8)$\\
\hline
$(5, 4)$ & $32$ & $10\ (16)$ & $3\ (16)$ & & & & & \\
\hline
$(6, 4)$ & $44$ & $16$ & $7\ (16)$ & $3\ (16)$ & & & & \\
\hline
$(7, 4)$ & $76$ & $22$ & $15\ (16)$ & $5\ (16)$ & $3\ (16)$ & & & \\
\hline
$(8, 4)$ & $146$ & $38$ & $16$ & $12\ (16)$ & $5\ (16)$ & $3\ (16)$ & & \\
\hline
$(9, 4)$ & $289$ & $72\ (73)$ & $19$ & $16$ & $10\ (16)$ & $4\ (16)$ & $3\ (16)$ & \\
\hline
$(10, 4)$ & $475$ & $140\ (144)$ & $33$ & $16$ & $16$ & $9\ (10)$ & $4\ (16)$ & $3\ (16)$\\
\hline
$(6, 5)$ & $89$ & $31\ (32)$ & $8\ (32)$ & $3\ (32)$ & & & & \\
\hline
$(7, 5)$ & $153$ & $45$ & $20\ (32)$ & $6\ (32)$ & $3\ (32)$ & & & \\
\hline
$(8, 5)$ & $292$ & $76$ & $32$ & $16\ (32)$ & $5\ (32)$ & $3\ (32)$ & & \\
\hline
$(9, 5)$ & $585$ & $145\ (146)$ & $32$ & $32$ & $13\ (32)$ & $5\ (32)$ & $3\ (32)$ & \\
\hline
$(7, 6)$ & $307$ & $90$ & $28\ (64)$ & $7\ (64)$ & $3\ (64)$ & & & \\
\hline
$(8, 6)$ & $585$ & $153$ & $64$ & $20\ (64)$ & $6\ (64)$ & $3\ (64)$ & & \\
\hline
$(9, 6)$ & $1170$ & $290\ (292)$ & $64$ & $51\ (64)$ & $17\ (64)$ & $5\ (64)$ & $3\ (64)$ & \\
\hline
$(8, 7)$ & $1170$ & $306\ (307)$ & $112\ (128)$ & $26\ (128)$ & $6\ (128)$ & $3\ (128)$ & & \\
\hline
$(9, 7)$ & $2340$ & $581\ (585)$ & $153$ & $74\ (128)$ & $22\ (128)$ & $6\ (128)$ & $3\ (128)$ & \\
\hline
$(9, 8)$ & $4681$ & $1164\ (1170)$ & $307$ & $96\ (256)$ & $24\ (256)$ & $6\ (256)$ & $3\ (256)$ & \\
\hline
\end{tabular}
\end{threeparttable}
\end{table*}
}

\section{Quantum enumerators for CWS codes}
\label{sec:quantum_enum}
In this section, we discuss the interpretations for the  $A$-type Shor-Laflamme weight enumerator $\{A_{j}\}$   in Def.~\ref{def:SL} and the quantum shadow enumerator $\{D_j\}$ in Def.~\ref{def:shadow} for a CWS code.
 
For $W\subset\cP^n$ and $\beta\in\cP^n$, define
\begin{align*}
    W_{\textrm{c}}(\beta) =& \lvert \{\omega\in W : \omega\beta=\beta \omega\}\rvert,\\
    W_{\textrm{a}}(\beta) =& \lvert \{\omega\in W :\omega\beta = -\beta \omega\}\rvert.
\end{align*}
The following theorem provides  an interpretation of $\{A_{j}\}$ 
in terms of the commutation relations between the CWS group and word operators of a CWS code.
\begin{theorem}
\label{thm:A}
Let $\mathcal{Q}$ be an  $((n, M, d; 0))$ CWS code with CWS group $S$ and a set of word operators $W$. Then,
    \begin{equation}
        A_{j} = \frac{1}{M^{2}}\sum_{\beta \in \mathcal{E}_{j}^{n} \cap S} \left\lvert (W_{\textrm{c}}(\beta) - W_{\textrm{a}}(\beta)) \right\rvert^{2}.
    \end{equation}
    Moreover, for $0 \leq j < d$, we have
    \begin{equation}
        A_{j} = \left\lvert \{\beta \in \mathcal{E}_{j}^{n} \cap S :  W_{\textrm{c}}(\beta) = M\} \right\rvert.
    \end{equation}
\end{theorem}

\begin{proof}
    Let $| S \rangle$ be the stabilize state of $S$ and $W = \{\omega_{k}\}_{k=1}^M$ with $\omega_1=I^{\otimes n}$.
    A basis for $\cQ$ is $\{\omega_k \ket{S}\}_{k=1}^M$.
    Then,  by Lemma~\ref{lemma:SL_enum},
    \begin{align*}
        A_{j} =& \frac{1}{M^{2}}\sum_{\beta \in \mathcal{E}_{j}^{n}} \left\lvert\sum_{k=1}^M  \langle S | \omega_{k}^{\dagger} \beta \omega_{k}| S \rangle \right\rvert^{2}\\
        =& \frac{1}{M^{2}}\sum_{\beta \in \mathcal{E}_{j}^{n}} \left\lvert (W_{\textrm{c}}(\beta) - W_{\textrm{a}}(\beta)) \langle S| \beta| S\rangle\right\rvert^{2}\\
        =& \frac{1}{M^{2}}\sum_{\beta \in \mathcal{E}_{j}^{n} \cap S} \left\lvert (W_{\textrm{c}}(\beta) - W_{\textrm{a}}(\beta)) \right\rvert^{2},
    \end{align*}
    where the second equality is because any two Pauli operators either commute or anticommute.

   By definition, any $\beta \in \mathcal{E}^{n}$ of weight less than the minimum distance $d$ is detectable by $\cQ$.  
      For $0\leq j < d$, the QEC condition~(\ref{error_con_1}) implies for  $\beta \in \mathcal{E}_j^{n}$, 
      $\langle S | \omega_{k}^{\dagger} \beta \omega_{k}| S \rangle= \langle S |   \beta  | S \rangle$ for all $k$ and  thus $\omega_{k}^{\dagger} \beta \omega_{k}=\beta$. Consequently, for $0\leq j < d$

    \begin{equation*}
        A_{j} = \left\lvert \{\beta \in \mathcal{E}_{j}^{n}: \beta \in S \text{ and } W_{\textrm{c}}(\beta) = M\} \right\rvert.
    \end{equation*}

\end{proof}


Next, let us consider the quantum shadow enumerators~\cite{R99}.
In Def.~\ref{def:shadow}, the shadow enumerator involves 
the complex conjugate of the code projector. 
For CWS-type codes,  one can observe  that 
the conjugate of a CWS state and the CWS state differ by only  a Pauli operator. Using the canonical representation of stabilizer states \cite{DDM03,VDN10},
we know that the conjugate state of a stabilizer state $\ket{S}$ is of the form:
\begin{align}\overline{|S \rangle} = \zeta \ket{S}\end{align}
where  $\zeta$  is an $n$-fold Pauli operator that is  the tensor product of $I$'s and $\sigma_z$'s.

\begin{theorem}
    \label{lem:shadow_main}
    Let $\mathcal{Q}$ be an $((n, M, d; 0))$ CWS code with  CWS group $S$ and a set of word operators $W$. Let $|S \rangle$ be the stabilizer state of $S$ and  $\overline{|S  \rangle}$ be its complex conjugate. Suppose that $\overline{|S \rangle} = \zeta \ket{S}$ for some  $\zeta\in\cP^n$. Then, the quantum shadow enumerator $\{D_j\}_{j=1}^n$ for $\mathcal{Q}$ is given by
    \begin{equation}
        D_{j} = \left\{
        \begin{aligned}
            &\frac{f(j)}{M}, & &\text{if }  \sigma_{y}^{\otimes n}\zeta \in \pm S;\\
            &f(j), & & \text{otherwise},
        \end{aligned}
        \right.
    \end{equation}
    where 
    \begin{align*}
        f(j) =& \frac{\lvert\{(\alpha, \beta) \in W S \times W\sigma_{y}^{\otimes n} \zeta S : \wt{\alpha\beta}= j \}\rvert}{\lvert S \rvert}.
        \end{align*}
\end{theorem}
\begin{proof}
Let $P$ be the projector onto $\cQ$ and 
$\hPQ$ be given as in Def.~\ref{def:shadow}.
Let $W = \{\omega_{k}\}_{k=1}^M$. Then, $P=\sum_{k=1}^M \omega_{k} \ket{S}\bra{S}\omega_{k}^{\dagger}$  and 
 $\hPQ=\sum_{k=1}^M \sigma_{y}^{\otimes n}\overline{\omega}_{k} \overline{\ket{S}} \overline{\bra{S}}\overline{\omega}_{k}^{\dagger} \sigma_{y}^{\otimes n}$.
After some calculations, we have
\begin{align*}
    {\rm Tr}(\PQ\hPQ) =& \sum_{j, k} \langle S| w_{j}^{\dagger} w_{k} | S \rangle \langle S| w_{k}^{\dagger} \left(\sum_{i} w_{i}|S \rangle \langle S| w_{i}^{\dagger}\right) \sigma_{y}^{\otimes n}\overline{\omega}_{k} \\
    &\cdot \overline{\ket{S}}\overline{\bra{S}}\overline{\omega}_{k}^{\dagger} \sigma_{y}^{\otimes n} w_{j} | S \rangle \\
    =& \sum_{k} \langle S| w_{k}^{\dagger} \sigma_{y}^{\otimes n}\overline{\omega}_{k}\overline{\ket{S}} \overline{\bra{S}}\overline{\omega}_{k}^{\dagger} \sigma_{y}^{\otimes n} w_{k} | S \rangle\\
    =& \sum_{k} \left\lvert\langle S| w_{k}^{\dagger} \sigma_{y}^{\otimes n}\overline{\omega}_{k}\overline{\ket{S}} \right\rvert^{2}\\
    =& \sum_{k} \left\lvert\langle S| \sigma_{y}^{\otimes n}\overline{\ket{S}} \right\rvert^{2}\\
    =& M \left\lvert\langle S| \sigma_{y}^{\otimes n}\zeta \ket{S} \right\rvert^{2}.
\end{align*}

If $ \left\lvert\langle S| \sigma_{y}^{\otimes n} \zeta \ket{S} \right\rvert^{2} \neq 0$,  by ~(\ref{eq:Phat}), we have 
\begin{align*}
    D_{j}  =& \frac{1}{M \left\lvert\langle S| \sigma_{y}^{\otimes n}\zeta \ket{S} \right\rvert^{2}}\sum_{\beta \in \mathcal{E}_{j}^{n}}{\rm Tr}(\sigma_{y}^{\otimes n} \beta \PQ \beta \sigma_{y}^{\otimes n}  \overline{P}_{\mathcal{Q}})\\
    =& \frac{1}{M \left\lvert\langle S| \sigma_{y}^{\otimes n}\zeta \ket{S} \right\rvert^{2}}\sum_{\beta \in \mathcal{E}_{j}^{n}}\sum_{i, k} \overline{\langle S|} \overline{w}_{i}^{\dagger} \sigma_{y}^{\otimes n} \beta w_{k}|S \rangle\\
    & \cdot \langle S| w_{k}^{\dagger} \beta \sigma_{y}^{\otimes n} \overline{w}_{i}  \overline{| S\rangle}\\
    =& \frac{1}{M \left\lvert\langle S| \sigma_{y}^{\otimes n}\zeta \ket{S} \right\rvert^{2}}\sum_{\beta \in \mathcal{E}_{j}^{n}}\sum_{i, k} \left\lvert \overline{\langle S|} \overline{w}_{i}^{\dagger} \sigma_{y}^{\otimes n} \beta w_{k}|S \rangle \right\rvert^{2}\\
    =& \frac{1}{M \left\lvert\langle S| \sigma_{y}^{\otimes n}\zeta \ket{S} \right\rvert^{2}}\sum_{\beta \in \mathcal{E}_{j}^{n}}\sum_{i, k} \left\lvert \langle S| \zeta^{\dagger}\overline{w}_{i}^{\dagger} \sigma_{y}^{\otimes n} \beta w_{k}|S \rangle \right\rvert^{2}.
\end{align*}
Recall that for $\omega \in \cE^{n}$,
\begin{equation*}
    \lvert \langle S|\omega|S \rangle \rvert = \left\{
    \begin{aligned}
    &1, & &\text{if } \omega \in \pm S;\\
    &0, & &\text{otherwise}.
    \end{aligned}
    \right.
\end{equation*}
Thus,
\begin{equation*}
    \left\lvert \langle S| \zeta^{\dagger} \overline{w}_{i}^{\dagger} \sigma_{y}^{\otimes n} \beta w_{k}|S \rangle \right\rvert = \left\{
    \begin{aligned}
        &1, & &\text{if } \beta \in \pm  w_{i} w_{k} \sigma_{y}^{\otimes n} \zeta S;\\
        &0, & &\text{otherwise}.
    \end{aligned}
    \right.
\end{equation*}
Then, we can rewrite $D_{j}$ as
\begin{align*}
    D_{j} 
    =& \frac{f(j)}{M}.
\end{align*}
If $\left\lvert\langle S| \sigma_{y}^{\otimes n}\overline{\ket{S}} \right\rvert^{2} = 0$, $D_{j} = \sum_{\beta \in \mathcal{E}_{j}^{n}}{\rm Tr}(\sigma_{y}^{\otimes n} \beta \PQ \beta \sigma_{y}^{\otimes n}  \overline{P}_{\mathcal{Q}})$. Thus, we have
\begin{equation*}
    D_{j} = f(j).
\end{equation*}
\end{proof}

Therefore, the shadow enumerator $\{D_j\}$ is the distance enumerator between $WS$ and $\big(\zeta \sigma_y^{\otimes n}\big) WS$.
 
 \begin{corollary}
    Let $\mathcal{Q}$ be an $((n, M, d; 0))$ CWS code with   CWS group $S$ and a set of word operators $W$. Let $|S \rangle$ be the stabilizer state of $S$ and $\overline{| S \rangle}$ be its complex conjugate. 
    If $\overline{|S \rangle} = {\bm{i}}^k|S \rangle$ for $k \in \{0,1,2,3\}$ and $\sigma_{y}^{\otimes n} \in WS$, then
    \begin{equation}
        D_{j} = \left\{
        \begin{aligned}
            &B_{j}, & & \text{if } \sigma_{y}^{n} \in \pm S;\\
            &\frac{d'(\alpha, \beta)}{\lvert S \rvert}, & &\text{otherwise,}
        \end{aligned} \right.
    \end{equation}
    where
    \begin{equation*}
        d'(\alpha, \beta) = \lvert\{(\alpha, \beta) \in W S \times W\sigma_{y}^{\otimes n}S : \wt{\alpha\beta}= j \}\rvert
    \end{equation*}
    and $\{B_{j}\}$ is the $B$-type SL weight enumerator for $\cQ$.
\end{corollary}

\begin{example}
Consider the $((5, 2, 3; 0))$ stabilizer code \cite{NC00} with stabilizer group generated by $\sigma_{x} \otimes \sigma_{z} \otimes \sigma_{z} \otimes \sigma_{x} \otimes I$ and its cyclic shifts. The logical operators are $\sigma_{x}^{\otimes 5}$ and $\sigma_{z}^{\otimes 5}$. 
    
    As a CWS code, its CWS group $S$ includes the  stabilizers and  $\sigma_{z}^{\otimes 5}$ and the word operators are 
    $W = \{ I^{\otimes 5}, \sigma_{x}^{\otimes 5}\}$. 
    The state $\ket{S}$ is  as follows:
    \begin{equation*}
    \begin{aligned}
        |S \rangle = \frac{1}{4}\{&|00000\rangle + |10010\rangle +|01001\rangle + |10100\rangle\\
        &+ |01010\rangle -|11011\rangle -|00110\rangle -|11000\rangle\\
        &- |11101\rangle - |00011\rangle - |11110\rangle - |01111\rangle \\
        &- |10001\rangle - |01100\rangle -|10111\rangle +|00101\rangle\}.
    \end{aligned}
    \end{equation*}
    In this case, $\overline{|S\rangle} = |S \rangle$ and $\sigma_{y}^{\otimes n} \in \{\pm 1, \pm \bm{i}\} \times S$. 
    The Shor-Laflamme enumerator ${A_{j}}$ for the $((5, 2, 3; 0))$ code corresponds to the weight enumerator for the stabilizer group. We can directly count the weight enumeration to obtain the values of ${A_{j}}$.
Alternatively, we can use Theorem \ref{thm:A} to calculate ${A_{j}}$. Both methods provide us with the same results:
    \begin{equation*}
        A_{0} = 1, A_{1} = 0, A_{2} = 0, A_{3} = 0, A_{4} = 15, A_{5} = 0.
    \end{equation*}

    For the shadow enumerator $\{D_{j}\}$, we can use Lemma \ref{lem:shadow} or Theorem \ref{lem:shadow_main}
    and both methods give us
    \begin{equation*}
        D_{0} = 1, D_{1} = 0, D_{2} = 0, D_{3} = 30, D_{4} = 15, D_{5} = 18.
    \end{equation*}

\end{example}

    \begin{example}
Consider   the $((5, 6, 2; 0))$ CWS code \cite{RHSS97}, which can be formulated as a CWS code \cite{CSSZ09}. Its CWS group $S$ is generated by $\sigma_{z} \otimes \sigma_{x} \otimes \sigma_{z} \otimes I \otimes I$ and its cyclic shifts. The word operators are $I_{2^{5}}$, $\sigma_{z} \otimes \sigma_{z} \otimes I \otimes \sigma_{z} \otimes I$ and its cyclic shifts. The stabilizer state $\ket{S}$ is as follows:
    \begin{equation*}
    \begin{aligned}
        |S \rangle = \frac{1}{4\sqrt{2}}\{&|00000\rangle + |00111\rangle + |10011\rangle + |11001\rangle\\
        &+ |11100\rangle + |01110\rangle + |00101\rangle + |10010\rangle \\
        &+ |01001\rangle + |10100\rangle + |01010\rangle + |10000\rangle\\
        &+ |01000\rangle + |00100\rangle + |00010\rangle + |00001\rangle\\
        &- |00011\rangle - |10001\rangle - |11000\rangle - |01100\rangle\\
        &- |00110\rangle - |01111\rangle - |10111\rangle - |11011\rangle\\
        &- |11101\rangle - |11110\rangle - |10101\rangle - |11010\rangle\\
        &- |01101\rangle - |10110\rangle - |01011\rangle - |11111\rangle\}.
    \end{aligned}
    \end{equation*}
    In this case, $\overline{|S\rangle} = |S \rangle$ and $\sigma_{y}^{\otimes n} \notin \pm S$. Using Theorem \ref{thm:A} and Theorem \ref{lem:shadow_main}, we have
    \begin{equation*}
        A_{0} = 1, A_{1} = 0, A_{2} = 0, A_{3} = 0, A_{4} = \frac{5}{3}, A_{5} = \frac{8}{3},
    \end{equation*}
    and
    \begin{equation*}
        D_{0} = 0, D_{1} = 30, D_{2} = 60, D_{3} = 360, D_{4} = 420, D_{5} = 282.
    \end{equation*}
    Moreover, using Lemma \ref{lem:CWS_exp}, we  have
    \begin{equation*}
        B_{0} = 1, B_{1} = 0, B_{2} = 20, B_{3} = 50, B_{4} = 75, B_{5} = 46.
    \end{equation*}
    Alternatively, we can apply the MacWilliams identities and Lemma \ref{lem:shadow} to calculate the values of $\{A_{j}\}$ and $\{D_{j}\}$ and obtain the same answers.

\end{example}

\section{conclusions}
\label{sec:con}

In this paper, we have developed semidefinite programs for EA-CWS quantum codes to establish upper bounds on code size based on code length and minimum distance. These semidefinite constraints are rooted in triple distance relations within a set of undetectable errors. Significantly, our semidefinite constraints improve upon several upper bounds derived solely from linear constraints.

We have also shown that the semidefinite constraints can be strictly more powerful than the linear constraints derived from the MacWilliams identities. However, at present, we have not observed any improvements in the SDP bounds over the LP bounds for standard CWS codes or stabilizer codes, primarily because of the compensatory effect of the linear shadow identities.

Additionally, we introduced the concept of GS sets, expanding upon the notion of stabilizers within stabilizer codes. For EA-CWS codes, we have shown that the spectrum must be confined to a subset of ${0, \pm 1}$, enhancing our comprehension of quantum code structures.

The minimum distance of a quantum code can be upper bounded by the minimum weight of a set of undetectable errors. In CWS-type scenarios, we selected the word operators and the isotropic subgroup to define two sets, $WS_{I}$ and $S_{I}$, utilizing their specific distance enumerations to establish upper bounds on the  code size. However, it is worth exploring if there are more suitable candidates for this characterization.

Furthermore, the approach we employed to construct the semidefinite constraints hints at the potential for a generalized version of the SDP method in classical cases. This includes SDPs based on distances between multiple points \cite{GMS12} and semidefinite constraints derived from a split Terwilliger algebra \cite{TLY23}.

Both the linear constraints and semidefinite constraints are closely related to our understanding of weight enumerators. Thus we also delved into combinatorial explanations of the weight enumerator ${A_{j}}$ and the quantum shadow enumerators for CWS codes. Regarding the quantum shadow enumerators, they can be viewed as  the distance enumerator between  the set $WS$ and one of its coset. This structural insight may open up opportunities to establish additional constraints.

The MATLAB programs used to obtain the numerical results presented in this paper can be accessed at:
\begin{center}
\url{https://github.com/PinChiehTseng/SDP_quantum_code}
\end{center}

\begin{IEEEbiographynophoto}{Ching-Yi Lai}
	(M'14-SM'22)  received his MS in 2006 and BS in 2004 from National Tsing Hua University, Taiwan and completed his Ph.D. at the University of Southern California in 2013. Following   postdoctoral   positions at University of Technology Sydney and Academia Sinica, he joined the Institute of Communications Engineering at National Yang Ming Chiao Tung University, Hsinchu, Taiwan. In 2022, Dr. Lai was promoted to Associate Professor. 
Dr. Lai was honored with the Young Scholar Fellowship from the National Science and Technology Council, Taiwan. Additionally, he received the 2023 IEEE Information Society Taipei Chapter and IEEE Communications Society Taipei/Tainan Chapter Best Paper Award for Young Scholars. His research explores diverse areas, including quantum coding theory, quantum information theory, fault-tolerant quantum computation, and quantum cryptography.
\end{IEEEbiographynophoto}

\begin{IEEEbiographynophoto}{Pin-Chieh Tseng}
	received his MS in 2020 from National Tsing Hua University, Taiwan and BS in 2017 from National Cheng Kung University, Taiwan and is current a Ph.d student at National Yang Ming Chiao Tung University, Taiwan.
\end{IEEEbiographynophoto}

\begin{IEEEbiographynophoto}{Wei-Hsuan Yu}
 received his MS from National Taiwan Normal University and BS from National Tsing Hua University, Taiwan and completed his Ph.D. at the University of Maryland, College Park in 2014. Following postdoctoral positions at Michigan State University and ICERM at Brown University, he joined the Department of Mathematics at National Central University, Taoyuan, Taiwan in 2018. In 2021, Dr. Yu was promoted to Associate Professor. 
Dr. Yu was honored with the 4years Young Scholar Grant (2020-24) from the National Science and Technology Council, Taiwan. Additionally, he received the 2020 National Center for Theoretical Sciences (NCTS) Young Theorist Award. His research explores in discrete geometry and combinatorics. 
\end{IEEEbiographynophoto}
\end{document}